\title{\huge \textbf{Unbiased estimation in two-stage adaptive enrichment designs}}

\author{Enyu Li$^{1\ast}$, Nigel Stallard$^{1}$, Ekkehard Glimm$^{2}$, and Peter K. Kimani$^{1}$ \\[8pt]
\textit{$^1$Clinical Trials Unit, University of Warwick, Coventry, U.K.} \\
\textit{$^2$Advanced Methodology and Data Science, Novartis Pharma AG, Basel, Switzerland}
\\[8pt]
$^\ast$Email: \href{mailto:enyu.li@warwick.ac.uk}{Enyu.Li@warwick.ac.uk}}

\date{ }

\documentclass[hidelinks, 11pt]{article}

\usepackage[margin=0.78in]{geometry}
\usepackage{hyperref}
\usepackage{amssymb, amsmath}
\usepackage{aliascnt}
\usepackage{graphicx}
\usepackage{booktabs}
\usepackage{multirow}
\usepackage{bm}
\usepackage{caption}
\usepackage{subcaption}
\usepackage{natbib}
\usepackage{pdfpages}
\usepackage{comment}
\usepackage{romannum} 
\usepackage{amsthm}
\usepackage{dsfont}
\usepackage{bbm}
\usepackage{enumitem}
\usepackage{tikz}
\usepackage{float}
\usepackage{array}
\usepackage{makecell}
\usepackage[nameinlink, capitalise]{cleveref}
\usepackage{xcolor}

\definecolor{myblue}{RGB}{0, 114, 178}
\definecolor{myred}{RGB}{213, 94, 0}
\definecolor{mygreen}{RGB}{0, 158, 115}
\definecolor{myindigo}{RGB}{63, 81, 181}
\definecolor{mycoral}{RGB}{255, 112, 67}
\definecolor{mygold}{RGB}{255, 202, 40}
\definecolor{myteal}{RGB}{38, 166, 154}
\definecolor{mypurple}{RGB}{126, 87, 194}
\definecolor{myorange}{RGB}{255, 167, 38}

\newcolumntype{C}[1]{>{\centering\arraybackslash}m{#1}}

\newtheorem{theorem}{Theorem}

\newaliascnt{lemma}{theorem}

\aliascntresetthe{lemma}

\newaliascnt{corollary}{theorem}
\newtheorem{corollary}[corollary]{Corollary}
\aliascntresetthe{corollary}

\newaliascnt{definition}{theorem}
\newtheorem{definition}[definition]{Definition}
\aliascntresetthe{definition}

\crefname{theorem}{Theorem}{Theorems}
\crefname{lemma}{Lemma}{Lemmas}
\crefname{corollary}{Corollary}{Corollaries}
\crefname{definition}{Definition}{Definitions}
\crefname{equation}{Equation}{Equations}
\crefname{figure}{Figure}{Figures}

\newcommand{\E}{\mathbb{E}}

\newcommand{\pr}[1]{\textup{(}\!#1\/\textup{)}}
\newcommand{\dif}{\,\mathrm{d}}

\hypersetup{
	colorlinks=true,  
	citecolor=blue,        
	linkcolor=black,        
	urlcolor=black          
}

\usetikzlibrary{arrows.meta,
	arrows,
	chains,
	matrix,
	positioning,
	scopes,
	shapes.geometric,
	shadows.blur,
	shapes.symbols,
	decorations.pathreplacing
}

\tikzstyle{block} = [rectangle, draw, fill=white!20, text width=8em, text centered, rounded corners, minimum height=3.5em, line width=0.3mm]
\tikzstyle{decision} = [rectangle, draw, fill=white!20, text width=8em, text centered, rounded corners, node distance=1.5cm and -0.3cm, line width=0.3mm, minimum height=3em]
\tikzstyle{line} = [draw, -Latex]

\begin{document}
\pagenumbering{arabic}
\maketitle

\begin{abstract}

Recent advances in biomedical research have identified an increasing number of biomarkers associated with heterogeneity in patient responses to medical treatments. When a treatment is suspected to benefit certain patient subpopulations, adaptive enrichment designs may be more efficient and ethical. In such designs, an interim analysis is incorporated during the trial to select patient subpopulations for which the experimental treatment appears promising, according to predefined subpopulation selection rules. However, data-dependent selection can induce selection bias, causing conventional maximum likelihood estimators (MLEs) to overestimate the treatment effect in the selected patient subgroup. Existing inference methods for addressing this bias are typically rule-specific, highlighting the need for an estimation framework that accommodate a broader class of subpopulation selection rules. In this work, we define a general class of subpopulation selection rules based on the sample space partition condition and provide a systematic derivation that yields a unified formula for the Uniformly Minimum Variance Conditional Unbiased Estimator (UMVCUE). This generality allows our formulation to encompass a wide spectrum of adaptive enrichment designs, eliminating the necessity for case-specific derivations for each new design. Extensive simulations confirm the unbiasedness of the proposed UMVCUE, ensuring that therapeutic benefits are not overestimated. By bridging the gap between flexible interim subpopulation selection and rigorous statistical inference, our framework has the potential to facilitate the implementation of diverse subpopulation selection rules with greater ease in real-world trials and promote more efficient and ethical drug development.

\end{abstract}

\hspace{-1em} {\em Key words}: Adaptive enrichment design; Point estimation; Precision medicine; Seamless phase \Romannum{2}/\Romannum{3} design; Subgroup analysis, Subpopulation selection.

\section{Introduction} \label{sec:intro}

Recent developments in biomedical science have revealed an increasing number of biomarkers associated with heterogeneity in patient responses to medical treatments. For example, metastatic breast cancer patients whose tumours overexpress human epidermal growth factor receptor-2 (HER2) have been reported to experience greater benefit from HER2-targeted therapies \citep{baselga2001herceptin, capelan2013pertuzumab}. Treatment heterogeneity, however, need not be driven solely by genetic biomarkers. It may also be related to other baseline patient characteristics; for instance, \citet{kirsch2008initial} reported in a meta-analysis that the benefit of certain antidepressants appeared to be concentrated among patients with more severe depression at baseline. To account for such potential treatment heterogeneity, it is often more efficient and ethical to employ adaptive enrichment designs.  In such designs, an interim analysis is incorporated during the trial to select patient subpopulations for which the experimental treatment appears promising based on accrued data. Subsequently, enrolment is restricted to the selected subgroups. A prominent real-world application of adaptive enrichment design is the Early Minimally Invasive Removal of Intracerebral Hemorrhage (ENRICH) trial \citep{pradilla2024trial}.

In adaptive enrichment designs, subpopulation selection rules should ideally preserve the broadest possible beneficial patient group. For example, an optimal rule should proceed with the overall population if a general benefit is observed in stage 1, but adapt to a specific subpopulation when the treatment effect is localized. Various interim subpopulation selection rules have been proposed in the literature \citep{brannath2009confirmatory, wang2009adaptive, jenkins2011adaptive, rosenblum2011optimizing, rosenblum2013confidence, magnusson2013group,kimani2015estimation, rosenblum2015adaptive, kunzmann2017point, kimani2018point, simon2018using, kimani2020point, stefano2022comparison}.  However, such data-dependent adaptations to enrolment can introduce significant statistical challenges in drawing inference regarding treatment effect in the selected patient subpopulation. For instance, standard hypothesis testing procedures which ignore the adaptive nature of the design often leads to type I error inflation. Moreover, conventional maximum likelihood estimators (MLEs) often overestimate the treatment effect, and conventional confidence intervals may fail to achieve nominal coverage probabilities, due to selection bias induced by the population adaption.

\citet{european2007reflection} emphasised that, to use an adaptive design, statistical methods that control type I error and yield correct estimates for the treatment effect must be available. To address the inferential challenge in adaptive enrichment designs, an increasing number of statistical methods have been proposed, including \citet{brannath2009confirmatory, jenkins2011adaptive, rosenblum2011optimizing, magirr2013simultaneous, magnusson2013group, rosenblum2013confidence, kimani2015estimation, rosenblum2015adaptive,robertson2016accounting, kunzmann2017point, kimani2018point, simon2018using, kimani2020point}, and \citet{freidling2025selective}. For comprehensive and detailed reviews of statistical methods for general adaptive designs involving treatment selection, subpopulation selection, and early stopping, see \citet{bretz2009adaptive, robertson2023point, robertson2025confidence}. Currently, there are hypothesis testing procedures \citep{rosenblum2011optimizing, magirr2013simultaneous, rosenblum2015adaptive} controlling the family-wise type I error rate and confidence intervals methods \citep{magirr2013simultaneous, li2026confidence} guaranteeing the nominal coverage for adaptive enrichment designs with a general class of interim subpopulation selection rules, i.e., the applicability of these methods do not limit to design with certain subpopulation selection rules. In contrast, the majority of proposed point estimators, including \citet{kimani2015estimation, robertson2016accounting, kunzmann2017point, kimani2018point} and \cite{stefano2022comparison}, are tailored to specific subpopulation selection rules. This lack of general applicability may create barriers to adopting enrichment designs in practice. For example, to employ a new subpopulation selection rule, an unbiased estimator must be derived since estimators from current studies are rule-specific. Hence, there remains a pressing need for an estimation framework that can encompass a broader class of subpopulation selection rules.

A principled way to adjust for selection bias and obtain an unbiased estimator is to apply Rao-Blackwell theorem and Lehmann–Scheffé theorem. Given an unbiased estimator $U$ of the unknown parameter of interest $\theta$ and a complete sufficient statistic $T$, these two theorems imply that the conditional expectation $\E[U \mid T]$ is the uniformly minimum variance unbiased estimator (UMVUE) of $\theta$. In adaptive designs, however, it may be more appropriate to seek unbiased estimators conditional on the adaptations that have occurred. From this conditional perspective, \citet{cohen1989} provided the first derivation of the uniformly minimum variance conditional unbiased estimator (UMVCUE) in the setting of a two-stage adaptive design with treatment selection.

 Taking this conditional estimation perspective, we develop a general framework for deriving uniformly minimum variance conditional unbiased estimators (UMVCUEs) in two-stage adaptive enrichment designs. A key component of this framework is the definition of a broad class of subpopulation selection rules, which characterizes the stage 1 sample space partition induced by the interim selection decision and encompasses many selection rules proposed in the current literature. Rather than providing rule-specific estimators, our approach gives a unified and readily implementable expression of UMVCUEs for this entire class of designs.
 
\section{Problem Definition} \label{sec:problem}

\subsection{Setting and notation} \label{sec:notation}

Suppose the overall patient population $\mathcal{F}$ can be partitioned into $k$ disjoint subpopulations $\mathcal{S}_1, \dots, \mathcal{S}_k$, indexed by $\mathcal{K} = \{1,\dots, k\}$. For each subpopulation $\mathcal{S}_m \; (m \in \mathcal{K})$, let $p_m$ denote the proportion of patients from subpopulation $\mathcal{S}_m$ in the overall population. We assume $p_m$ is known for each $m \in \mathcal{K}$. For any arbitrary subset of indices $\mathcal{I} \subseteq \mathcal{K}$, we define the corresponding combined subpopulation $\mathcal{S}_\mathcal{I} = \bigcup_{m \in \mathcal{I}} \mathcal{S}_m$ and denote its proportion in the overall population by $p_\mathcal{I} = \sum_{m \in \mathcal{I}} p_m$.

We focus on adaptive enrichment designs with two stages. In stage 1, $n_1$ patients are enrolled from the overall population using stratified sampling based on the prevalence $p_m$ of each subpopulation, such that, for each $m \in \mathcal{K}$, $n_1p_m$ patients are enrolled from $\mathcal{S}_m$ in stage 1. At the interim analysis, a pre-specified subpopulation selection rule is used to identify a selected subpopulation $\mathcal{S}_{\mathcal{E}}$ with a promising observed stage 1 treatment effect estimate, where $\mathcal{E} \subseteq \mathcal{K}$ denotes the index set of individual subpopulations selected for stage two. In stage 2, $n_2$ patients are enrolled from $\mathcal{S}_\mathcal{E}$, with subpopulation sample size proportional to $p_m/p_\mathcal{E}$ for each $\mathcal{S}_m \; (m \in \mathcal{E})$. Within each enrolled subpopulation during a given stage, we assume that half of the patients are randomized to the treatment arm and half to the control arm. For the selected subpopulation $\mathcal{S}_{\mathcal{E}}$, let $r_{\mathcal{E}} = n_1 p_{\mathcal{E}}/n_2$ denote the ratio of the stage 1 sample size in $\mathcal{S}_{\mathcal E}$ to the stage 2 sample size. More generally, for any subpopulation $\mathcal{S}_{\mathcal{I}}$ with $\mathcal{I} \subseteq \mathcal{E}$, let $r_{\mathcal{I}} = n_1 p_{\mathcal I}/(n_2p_{\mathcal I}/ p_{\mathcal E}) $ denote the corresponding stage 1 to stage 2 sample size ratio. It follows that $r_{\mathcal{I}} = n_1p_{\mathcal E}/n_2  = r_{\mathcal{E}}$ for all $\mathcal{I} \subseteq \mathcal{E}$. We assume $n_1$ and $n_2$ are fixed at the beginning of the trial. In the confirmatory analysis, we are interested in the treatment effect in the selected subpopulation.

For each patient $i \in \{1, \dots, n_1+n_2\}$, we collect the following data, $(M_i, J_i, A_i, Y_i)$, where $M_i \in \mathcal{K}$ denotes the subpopulation index, $J_i \in \{1,2\}$ is the stage when the patient is enrolled, $A_i \in \{0,1\}$ is the treatment arm indicator ($0$ for control, $1$ for treatment), and $Y_i \in \mathbb{R}$ is the outcome. For a patient in subpopulation $\mathcal{S}_m$ assigned to treatment arm $a$, we assume that $Y_i \sim N(\mu_{ma}, \sigma^2)$, with common known variance $\sigma^2$. Let $\Delta_m = \mu_{m1} - \mu_{m0}$ be the true treatment effect in $\mathcal{S}_m$. For any combined subpopulation $S_\mathcal{I}$, the aggregated treatment effect is defined as $\Delta_{\mathcal{I}} = \sum_{m \in \mathcal{I}} p_m \Delta_m / p_\mathcal{I}$.

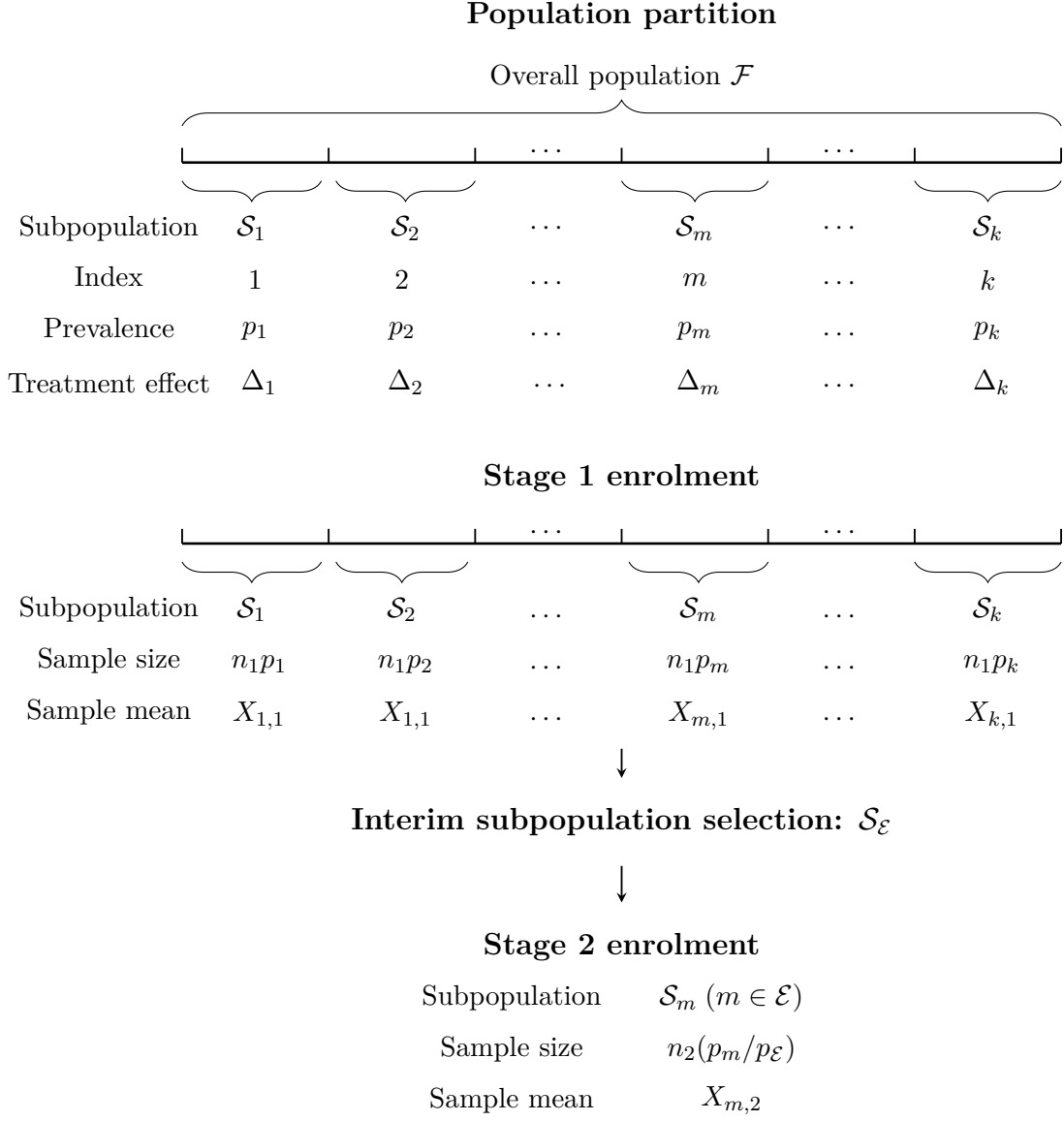
\begin{figure}[!htbp]
	\makebox[\linewidth][c]{
		\hspace*{-1cm}%
		\begin{tikzpicture}
			\node at (6, 6) {\large \textbf{Population partition}};
			
			\draw[line width=1pt] (0,4) -- (12,4);
			
			\draw[decorate,decoration={brace,amplitude=10pt}]
			(0,4.5) -- (12,4.5) node[midway,above=10pt] {Overall population $\mathcal{F}$};
			
			\foreach \x in {0,2,4,6,8,10,12} {
				\draw[line width=0.7pt] (\x,4) -- (\x,4.2);
			}
			
			\node at (9,4.15) {$\cdots$};
			\node at (5,4.15) {$\cdots$};
			
			\draw[decorate,decoration={brace,mirror,amplitude=8pt}]
			(0,3.75) -- (1.9,3.75) node[midway,below=10pt] {$\mathcal{S}_1$};
			\draw[decorate,decoration={brace,mirror,amplitude=8pt}]
			(2.1,3.75) -- (4,3.75) node[midway,below=10pt] {$\mathcal{S}_2$};
			\draw[decorate,decoration={brace,mirror,amplitude=8pt}]
			(6,3.75) -- (8,3.75) node[midway,below=10pt] {$\mathcal{S}_m$};
			\draw[decorate,decoration={brace,mirror,amplitude=8pt}]
			(10,3.75) -- (12,3.75) node[midway,below=10pt] {$\mathcal{S}_k$};
			\node at (9,3.1) {$\cdots$};
			\node at (5,3.1) {$\cdots$};
			
			\node at (-1,3.1) {Subpopulation};
			
			\node at (-1,2.45) {Index};
			\node at (1,2.4) {$1$};
			\node at (3,2.4) {$2$};
			\node at (5,2.35) {$\cdots$};
			\node at (7,2.4) {$m$};
			\node at (11,2.4) {$k$};
			\node at (9,2.35) {$\cdots$};
			
			\node at (-1,1.75) {Prevalence};
			\node at (1,1.7) {$p_1$};
			\node at (3,1.7) {$p_2$};
			\node at (5,1.65) {$\cdots$};
			\node at (7,1.7) {$p_m$};
			\node at (11,1.7) {$p_k$};
			\node at (9,1.65) {$\cdots$};
			
			\node at (-1,1) {Treatment effect};
			\node at (1.05, 1) {$\Delta_1$};
			\node at (3.05,1) {$\Delta_2$};
			\node at (5.05,0.95) {$\cdots$};
			\node at (7.05,1) {$\Delta_m$};
			\node at (11.05,1) {$\Delta_k$};
			\node at (9,0.95) {$\cdots$};

			\node at (6, -0.3) {\large \textbf{Stage 1 enrolment}};
			
			\draw[line width=1pt] (0,-1.2) -- (12,-1.2);
			
			\foreach \x in {0,2,4,6,8,10,12} {
				\draw[line width=0.7pt] (\x,-1.2) -- (\x,-1);
			}
			
			\node at (5,-1.05) {$\cdots$};
			\node at (9,-1.05) {$\cdots$};
			
			\draw[decorate,decoration={brace,mirror,amplitude=8pt}]
			(0,-1.45) -- (1.9,-1.45) node[midway,below=10pt] {$\mathcal{S}_1$};
			\draw[decorate,decoration={brace,mirror,amplitude=8pt}]
			(2.1,-1.45) -- (3.9,-1.45) node[midway,below=10pt] {$\mathcal{S}_2$};
			\draw[decorate,decoration={brace,mirror,amplitude=8pt}]
			(6.1,-1.45) -- (8,-1.45) node[midway,below=10pt] {$\mathcal{S}_m$};
			\draw[decorate,decoration={brace,mirror,amplitude=8pt}]
			(10,-1.45) -- (12,-1.45) node[midway,below=10pt] {$\mathcal{S}_k$};
			\node at (5,-2.2) {$\cdots$};
			\node at (9,-2.2) {$\cdots$};
			
			\node at (-1,-2.1) {Subpopulation};
			
			\node at (-1,-2.8) {Sample size};
			\node at (1.05,-2.85) {$n_1p_1$};
			\node at (3.05,-2.85) {$n_1p_2$};
			\node at (7.05,-2.85) {$n_1p_m$};
			\node at (11.05,-2.85) {$n_1p_k$};
			\node at (5,-2.9) {$\cdots$};
			\node at (9,-2.9) {$\cdots$};
			
			\node at (-1,-3.5) {Sample mean};
			\node at (1.05,-3.55) {$X_{1,1}$};
			\node at (3.05,-3.55) {$X_{1,1}$};
			\node at (7.05,-3.55) {$X_{m,1}$};
			\node at (11.05,-3.55) {$X_{k,1}$};
			\node at (5,-3.6) {$\cdots$};
			\node at (9,-3.6) {$\cdots$};
			
			\draw[-stealth, line width=0.7pt] (6,-4) -- (6,-4.4);
			
			\node at (6, -5) {\large \textbf{Interim subpopulation selection: $\mathcal{S}_\mathcal{E}$}};
			
			\draw[-stealth, line width=0.7pt] (6,-5.6) -- (6,-6.1);
			
			\node at (6, -6.7) {\large \textbf{Stage 2 enrolment}};
			\node at (4.5, -7.4) {Subpopulation};
			\node at (7.5, -7.4) {$\mathcal{S}_m\; (m \in \mathcal{E})$};
			\node at (4.5, -8.1) {Sample size};
			\node at (7.5, -8.1) {$n_2 (p_m/ p_\mathcal{E})$};
			\node at (4.5, -8.8) {Sample mean};
			\node at (7.5, -8.8) {$X_{m,2}$};
			
		\end{tikzpicture}
	}
	\caption{Illustration of the population partition, enrolment process, and notation}
	\label{fig:population}
\end{figure}

\subsection{Definition of the statistics} \label{sec:statistics}

For a subpopulation $\mathcal{S}_\mathcal{I} \; (\mathcal{I} \subseteq \mathcal{K})$ and a given stage $j$, provided that patients from each $\mathcal{S}_m \; (m \in \mathcal{I})$ are enrolled in stage $j$, we define the stage-wise sample mean difference between treatment and control as
\begin{equation*}
	X_{\mathcal{I}, j} = \frac{\sum_{\{i:M_i \in\mathcal{I}, J_i = j, A_i = 1\}}Y_i}{\vert \{i:M_i \in \mathcal{I}, J_i = j, A_i = 1\} \vert} - \frac{\sum_{\{i:M_i \in \mathcal{I}, J_i = j, A_i = 0\}}Y_i}{\vert \{i:M_i \in \mathcal{I}, J_i = j, A_i = 0\} \vert } .
\end{equation*}
Let $v^2_{\mathcal{I}, j} = 4 \sigma^2 / \vert \{i:M_i \in \mathcal{I}, J_i = j\} \vert$ denote the variance of $X_{\mathcal{I}, j}$. We have $X_{\mathcal{I}, j} \sim N(\Delta_{\mathcal{I}}, v^2_{\mathcal{I}, j})$. For $\mathcal{S}_\mathcal{I}$, if patients from each $\mathcal{S}_m \; (m \in \mathcal{I})$ are enrolled in both stages, we define the maximum likelihood estimator (MLE) for the treatment effect using the pooled data in both stages as
\begin{equation}
	X_{\mathcal{I}} = \frac{\sum_{\{i:M_i \in\mathcal{I}, A_i = 1\}}Y_i}{\vert \{i:M_i \in \mathcal{I}, A_i = 1\} \vert} - \frac{\sum_{\{i:M_i \in \mathcal{I}, A_i = 0\}}Y_i}{\vert \{i:M_i \in \mathcal{I}, A_i = 0\} \vert }.
	\label{MLE}
\end{equation}
Let $v^2_{\mathcal{I}} = 4 \sigma^2 / \vert \{i:M_i \in \mathcal{I}\} \vert$ denote the variance of $X_{\mathcal{I}}$. We have $X_{\mathcal{I}} \sim N(\Delta_{\mathcal{I}}, v^2_{\mathcal{I}})$. For notational simplicity, for an individual subpopulation $\mathcal{S}_m$, we write $X_{m, j}$ and $X_{m}$ instead of  $X_{\{m\}, j}$ and $X_{\{m\}}$, respectively, throughout the paper. We also write $X_{\mathcal{F}, j}$ and $X_{\mathcal{F}}$ instead of  $X_{\mathcal{K}, j}$ and $X_{\mathcal{K}}$. \cref{fig:population} summarises the population partition, enrolment process, and notation.

For any subpopulation $\mathcal{S}_{\mathcal I}$, we define $\boldsymbol{X}_{\mathcal I,1} = (X_{m,1})_{m \in \mathcal I}$ as the vector of stage 1 sample mean differences, and $\bm{\Delta}_{\mathcal I} = (\Delta_m)_{m \in \mathcal I}$ as the vector of true treatment effects, for the individual subpopulations comprising $\mathcal{S}_{\mathcal I}$. In both vectors, the elements are arranged in increasing order of the indices $m$. Throughout the paper, bold symbols are used to distinguish vectors from scalars.

For any arbitrary subset of indices $\mathcal{I} \subseteq \mathcal{K}$, we denote by $\mathcal{I}^{\prime}$ the complement of $\mathcal{I}$ in $\mathcal{K}$, so that $\mathcal{I}^{\prime} = \mathcal{K} \setminus \mathcal{I}$. Specifically, $\boldsymbol{X}_{\mathcal{E}^\prime, 1}$ denotes the stage 1 sample mean differences in all unselected individual subpopulations at the interim analysis.

\subsection{Interim subpopulation selection rule} \label{sec:class}
To formalize interim subpopulation selection, we define the pre-specified selection rule $D$ as a mapping from the stage 1 sample mean differences to a subset of $\mathcal{K}$, such that
\begin{equation*}
	D \colon \mathbb{R}^k \to \mathcal{P}(\mathcal{K}), \; \boldsymbol{x}_{\mathcal{K}, 1} \mapsto D(\boldsymbol{x}_{\mathcal{K}, 1}) = \mathcal{E},
\end{equation*}
where $\mathcal{P}(\mathcal{K})$ is the power set of $\mathcal{K}$, namely the set of all subsets of $\mathcal{K}$; $\boldsymbol{x}_{\mathcal{K}, 1}$ represents the vector of observed stage 1 sample mean differences in all $k$ individual subpopulations; and $\mathcal{E}$ denotes the index set of the selected subpopulation. This mapping dictates that enrolment in stage 2 will be enriched and restricted to the combined subpopulation $\mathcal{S}_\mathcal{E}$. Specifically, $D(\boldsymbol{x}_{\mathcal{K}, 1}) = \varnothing$ indicates that the trial stops for futility after stage 1.

Building on the class defined in \citet{li2026confidence}, we introduce a closely related class of interim subpopulation selection rules, denoted by $\mathcal{D}$, for which the theoretical results in \cref{sec:umvcue} provide a general derivation of the uniformly minimum variance conditional unbiased estimator (UMVCUE). We emphasise that $\mathcal{D}$ is introduced to characterize the class of selection rules $D$ for which our conditional point estimation framework is applicable, rather than as a prescriptive template for constructing such rules. For a given subpopulation selection rule $D$, let $\text{Im}(D) = \{ D(\boldsymbol{x}_{\mathcal{K}, 1}) : \boldsymbol{x}_{\mathcal{K}, 1} \in \mathbb{R}^K \}$ denote the image of $D$, which represents the set of all achievable selection outcomes. Subpopulation selection rules in $\mathcal{D}$ satisfy the following stage 1 sample space partition condition: for each non-empty realization of $\mathcal{E}$ where $\mathcal{E} \in \text{Im}(D) \setminus \{\varnothing\}$, there exist scalar-valued random thresholds $L_D(\mathcal{E}, \boldsymbol{X}_{\mathcal{E}^\prime, 1})$ and $U_D(\mathcal{E}, \boldsymbol{X}_{\mathcal{E}^\prime, 1})$ $(L_D(\mathcal{E}, \boldsymbol{X}_{\mathcal{E}^\prime, 1})  < U_D(\mathcal{E}, \boldsymbol{X}_{\mathcal{E}^\prime, 1}))$ such that
\begin{equation}
	\{\boldsymbol{X}_{\mathcal{K}, 1}: D(\boldsymbol{X}_{\mathcal{K}, 1}) =\mathcal{E}\} = \{\boldsymbol{X}_{\mathcal{K},1}: L_D(\mathcal{E}, \boldsymbol{X}_{\mathcal{E}^\prime, 1}) < X_{\mathcal{E}, 1} < U_D(\mathcal{E}, \boldsymbol{X}_{\mathcal{E}^\prime, 1})\}.
	\label{class:def}
\end{equation}
For a given subpopulation selection rule $D$, $L_D(\mathcal{E}, \boldsymbol{X}_{\mathcal{E}^\prime, 1})$ and $U_D(\mathcal{E}, \boldsymbol{X}_{\mathcal{E}^\prime, 1})$ are functionally determined by the realized index set $\mathcal{E}$ and the vector of stage 1 sample mean differences in all unselected subpopulations $\boldsymbol{X}_{\mathcal{E}^\prime, 1}$. $L_D(\mathcal{E}, \boldsymbol{X}_{\mathcal{E}^\prime, 1})$ and $U_D(\mathcal{E}, \boldsymbol{X}_{\mathcal{E}^\prime, 1})$ may also take values in $\overline{\mathbb{R}} = \mathbb{R} \cup \{-\infty,+\infty\}$. The inequalities may be taken as either strict or non-strict, as the distinction does not affect the validity of the proposed method. In general, under a selection rule $D \in \mathcal{D}$, a combined subpopulation $\mathcal{S}_\mathcal{E}$ is selected if and only if its stage 1 sample mean difference $X_{\mathcal{E}, 1}$ lies within the interval determined by $L_D(\mathcal{E}, \boldsymbol{X}_{\mathcal{E}^\prime, 1})$ and $U_D(\mathcal{E}, \boldsymbol{X}_{\mathcal{E}^\prime, 1})$. \cref{fig:class} provides a geometric illustration of the stage 1 sample space partition induced by the selection rule within class $\mathcal{D}$.

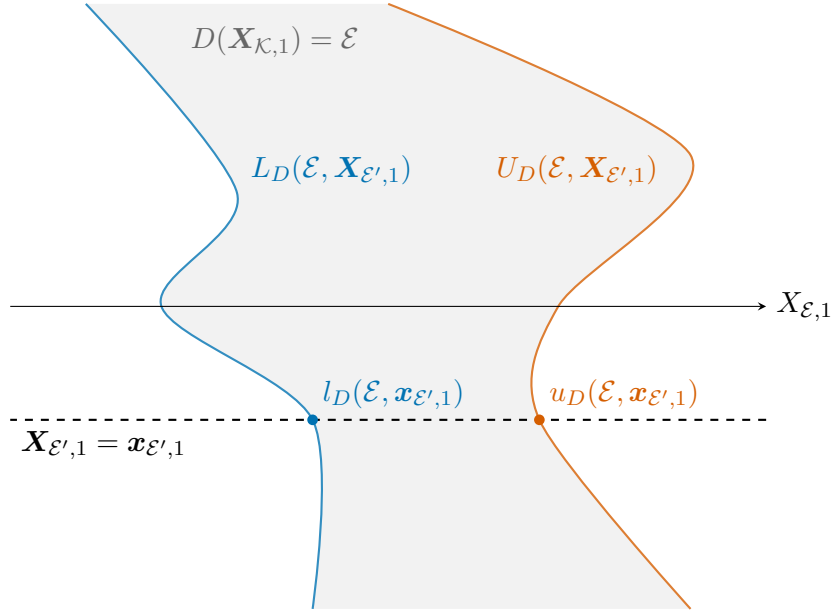
\begin{figure}[!htbp]
	\begin{center}
		\begin{tikzpicture}[>=stealth]
			\fill[black!5] 
			plot [smooth, tension=0.5] coordinates {(-4,4) (-2,1.5) (-3,0) (-1,-1.5) (-1,-4)} --
			(-1,-4) -- (4,-4) --
			plot [smooth, tension=0.5] coordinates {(4,-4) (2,-1.5) (2.25,0) (4,2) (0,4)};
			
			\node[black!60] at (-1.5, 3.5) {$D(\boldsymbol{X}_{\mathcal{K}, 1}) = \mathcal{E}$};
			
			\draw[myblue!80, thick, smooth, tension=0.5] plot coordinates { (-4,4) (-2,1.5) (-3,0) (-1,-1.5) (-1,-4) };
			\node[myblue] at (-0.75, 1.8) {$L_D(\mathcal{E}, \boldsymbol{X}_{\mathcal{E}^\prime, 1})$};
			
			\draw[myred!80, thick, smooth, tension=0.5] plot coordinates { (0,4) (4,2) (2.25,0) (2,-1.5) (4,-4) };
			\node[myred] at (2.5, 1.8) {$U_D(\mathcal{E}, \boldsymbol{X}_{\mathcal{E}^\prime, 1})$};
			
			\draw[->] (-5,0) -- (5,0) node[right] {$X_{\mathcal{E}, 1}$};
			
			\draw[black, thick, dashed] (5,-1.5) -- (-5,-1.5) node[below right] {$\boldsymbol{X}_{\mathcal{E}^\prime, 1}= \boldsymbol{x}_{\mathcal{E}^\prime, 1}$};
			\fill[myblue] (-1,-1.5) circle (2pt) node[above right] {$l_D(\mathcal{E}, \boldsymbol{x}_{\mathcal{E}^\prime, 1})$};
			\fill[myred] (2,-1.5) circle (2pt) node[above right] {$u_D(\mathcal{E}, \boldsymbol{x}_{\mathcal{E}^\prime, 1})$};
			
		\end{tikzpicture}
	\end{center}
	\caption{A geometric illustration of the stage 1 sample space partition characterized by class $\mathcal{D}$. The shaded region represents the stage 1 sample space where the subpopulation $\mathcal{S}_\mathcal{E}$ is selected, i.e., $D(\boldsymbol{X}_{\mathcal{K}, 1}) = \mathcal{E}$. For a fixed realization of $\mathbf{x}_{\mathcal{E}^\prime,1}$ (dashed line), the selection event is reduced to the interval $(l_D, u_D)$ for $X_{\mathcal{E},1}$.}
	\label{fig:class}
\end{figure}

The class $\mathcal{D}$ encompasses a broad spectrum of subpopulation selection rules commonly employed in adaptive enrichment designs, including those established in the literature \citep{rosenblum2011optimizing, rosenblum2013confidence, kimani2015estimation, robertson2016accounting, kimani2018point, kimani2020point}. For any selection rule in this class, the estimation framework proposed in \cref{sec:umvcue} can be used, through \cref{theorem}, to derive the UMVCUE of the treatment effect in the selected subpopulation.

Notably, although certain selection rules do not strictly conform to class $\mathcal{D}$, including those proposed by \citet{jenkins2011adaptive}, \citet{magnusson2013group}, \citet{rosenblum2015adaptive}, \citet{kunzmann2017point}, \citet{stefano2022comparison}, and \citet{freidling2025selective}, the sample space partition perspective proposed in this work remains applicable. By leveraging this perspective and the methodology presented in \cref{corollary2}, one can still obtain UMVCUEs for each individual subpopulation within the selected subpopulation and integrate a conditional unbiased estimator of the selected subpopulation, though this may not necessarily satisfy the UMVCUE property.

\section{Examples of subpopulation selection rules in class $\mathcal{D}$}

\subsection{A subpopulation selection rule for two subpopulations case} \label{sec:design1}

We first consider the two subpopulations case ($\mathcal{K} = \{1,2\}$) and the proposed subpopulation selection rule $D_1$ in \citet{li2026confidence} as illustrated in \cref{fig:decision1}. If the observed stage 1 sample mean difference in the overall population $\mathcal{F}$ exceeds the pre-specified threshold $\Delta_*$, the trial continues with $\mathcal{F}$ in stage 2. Otherwise, if the larger of the stage 1 sample mean differences in $\mathcal{S}_1$ and $\mathcal{S}_2$ exceeds $\Delta_*$, the trial enriches to the corresponding subpopulation. If none of these stage 1 sample mean differences exceeds $\Delta_*$, the trial stops for futility. \cref{fig:decision2} visualises the stage 1 sample space partition induced by $D_1$. 

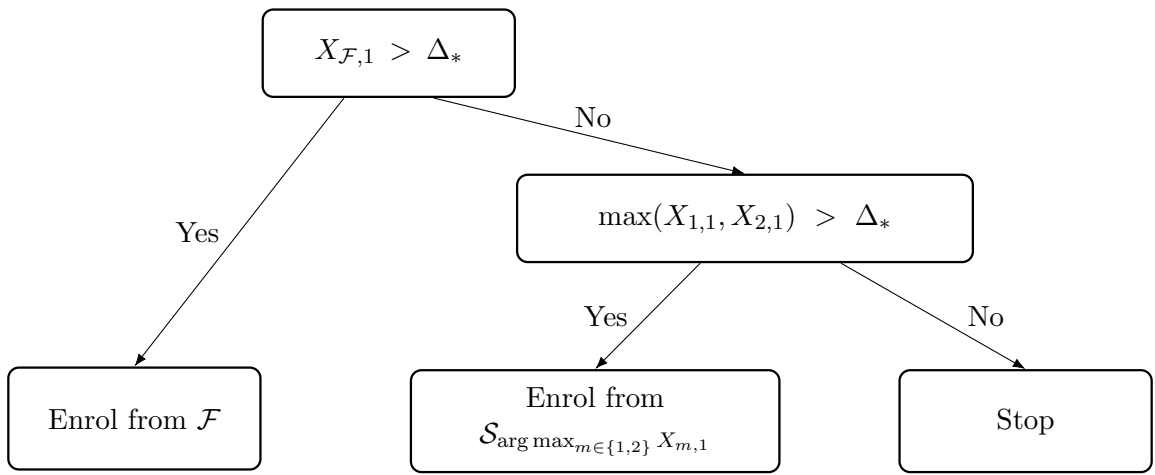
\begin{figure}[!htbp]
	\begin{center}
		\begin{tikzpicture}[node distance=1.5cm and 0cm, auto]

			\node (start) [decision] {$X_{\mathcal{F}, 1}> \Delta_*$};
			\node (outcome1) [block, below left=3.55cm and 0cm of start] {Enrol from $\mathcal{F}$};
			\node (dec2) [decision, below right=1cm and 0cm of start, text width=15em] {$\max(X_{1,1}, X_{2, 1}) > \Delta_*$};
			\node (outcome2) [block, below left=1.4cm and -3.5cm of dec2, text width=12em] {Enrol from $\mathcal{S}_{\arg\max_{m \in \{1,2\}}X_{m, 1}}$};
			\node (outcome3) [block, below right=1.4cm and -1cm of dec2] {Stop} ;
			
			\path [line] (start.-135) -- node [midway, left] {Yes \hspace{0.5em}} (outcome1.90);
			\path [line] (start.-45) -- node [pos=0.25, right] {\hspace{1.5em}  No} (dec2.90);
			\path [line] (dec2.-135) -- node [midway, left] {Yes \hspace{0.5em}} (outcome2.90);
			\path [line] (dec2.-25) -- node [midway, right] {\hspace{0.5em} No} (outcome3.90);
		\end{tikzpicture}
	\end{center}
	\caption{Schematic diagram of subpopulation selection rule $D_1$}
	\label{fig:decision1}
\end{figure}

\begin{figure}[htbp]
	\centering
	\begin{minipage}[c]{0.58\textwidth}
		\centering
			\begin{tikzpicture}[>=stealth, scale=1]
				\small
				\fill[myblue!15] (0,2) -- (2,0) -- (5,5) -- cycle;
				\fill[myblue!15] (2,0) -- (5,-3) -- (5,5) -- cycle;
				\fill[myblue!15] (0,2) -- (5,5) -- (-3,5) -- cycle;
				\fill[myred!20] (1,1) -- (1,-3) -- (5,-3) -- cycle;
				\fill[myindigo!20] (1,1) -- (-3,1) -- (-3,5) -- cycle;
				\fill[black!10] (-3,1) -- (1,1) -- (1,-3) -- (-3,-3) -- cycle;
				
				\draw[->] (-3,0) -- (5,0) node[right] {$X_{1, 1}$};
				\draw[->] (0,-3) -- (0,5) node[above] {$X_{2, 1}$};
				\draw[->, dashed] (-3,-3) -- (5,5) node[above] {$X_{\mathcal{F}, 1}$};
				
				\node[below left] at (0,0) {};
				
				\draw[black, thick, dashed] (5,-3) -- (-3,5) node[right] {};
				\draw[black, thick, dashed] (1,1) -- (1,-3) node[right] {};
				\draw[black, thick, dashed] (-3,1) -- (1,1) node[right] {};
				
				\fill (2,0) circle (1.5pt) node[above right] {\scriptsize $(\Delta_*/p_1,0)$};
				\fill (0,2) circle (1.5pt) node[above right] {\scriptsize $(0,\Delta_*/p_2)$};
				\fill (1,0) circle (1.5pt) node[below] {};
				\node at (1.1, 0) [below left] {\scriptsize $(\Delta_*,0)$};
				\fill (0,1) circle (1.5pt) node[above left] {\scriptsize $(0,\Delta_*)$};
				\fill (1,1) circle (1.5pt) node[right] {\scriptsize $(\Delta_*,\Delta_*)$};
				
				\node at (2, 4) {Continue with $\mathcal{F}$};
				\node at (2.5, -2.25) {Enrich to $\mathcal{S}_1$};
				\node at (-1.75, 2.35) {Enrich to $\mathcal{S}_2$};
				\node at (-2, -0.75) {Stop};
				
			\end{tikzpicture}
    \end{minipage}
	\begin{minipage}[c]{0.4\textwidth}
		\raggedright
			\begin{tikzpicture}[>=stealth, scale=1.1]
				\small
				\fill[myblue!15] (2,4) -- (5,4) -- (5,2) -- (2,2) -- cycle;
				\draw[->] (1,3) -- (5,3) node[right] {$X_{\mathcal{F}, 1}$};
				\node at (3.5, 3.75) {$\mathcal{E} = \{1,2\}, \; \mathcal{S}_\mathcal{E} = \mathcal{F}$};
				\node at (0.5, 4) {Continue with $\mathcal{F}$};
				
				\fill[myred!20] (2,1.5) -- (2,-0.5) -- (5,-0.5) -- cycle;
				\draw[->] (1,0.5) -- (5,0.5) node[right] {$X_{1, 1}$};
				\node at (3.4, -0.35) {$\mathcal{E} = \{1\}, \; \mathcal{S}_\mathcal{E} = \mathcal{S}_1$};
				\node at (0.5, 1.5) {Enrich to $\mathcal{S}_1$};
				
				\fill[myindigo!20] (2,-1) -- (2,-3) -- (5,-3) -- cycle;
				\draw[->] (1,-2) -- (5,-2) node[right] {$X_{2, 1}$};
				\node at (3.4, -2.85) {$\mathcal{E} = \{2\}, \; \mathcal{S}_\mathcal{E} = \mathcal{S}_2$};
				\node at (0.5, -1) {Enrich to $\mathcal{S}_2$};
			\end{tikzpicture}
     \end{minipage}
     \caption{A geometric illustration of stage 1 sample space partition induced by $D_1$. The shaded regions represent the stage 1 sample space where the overall population (blue), subpopulation 1 (orange), subpopulation 2 (purple), and stopping (grey) is selected.}
     \label{fig:decision2}
\end{figure}
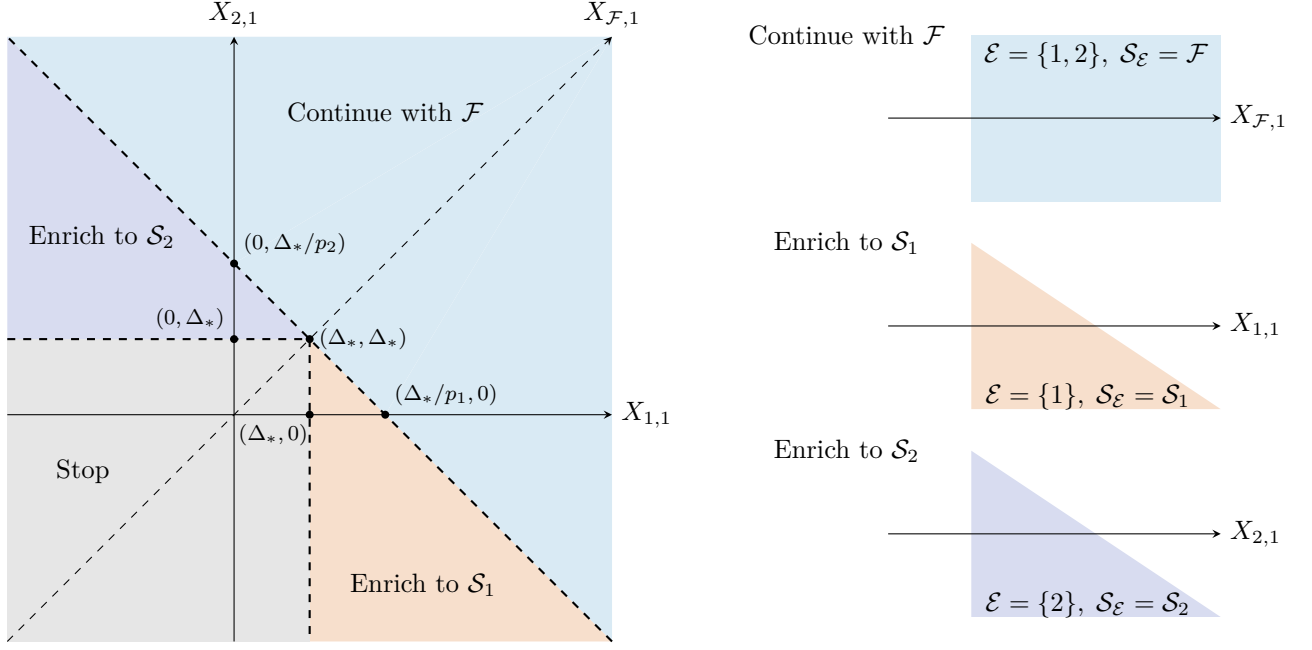

The possible selection outcomes of $D_1(\boldsymbol{X}_{\mathcal{K},1})$ are $\{1,2\}$, $\{1\}, \{2\}$, and $\varnothing$, corresponding respectively to continuation in $\mathcal{F}$, enrichment in $\mathcal{S}_1$, enrichment in $\mathcal{S}_2$, and stopping for futility. It is straightforward to verify that $D_1$ yields the following partition of the stage 1 sample space satisfying \cref{class:def} and therefore lies in class $\mathcal{D}$. The right panel of \cref{fig:decision2} visualises these partitions that are projected onto the stage 1 sample mean of the selected subpopulation.
\begin{align*}
	& \{\boldsymbol{X}_{\mathcal{K}, 1}: D_1(\boldsymbol{X}_{\mathcal{K}, 1}) =\{1,2\}\} = \{\boldsymbol{X}_{\mathcal{K},1}: \Delta_* < X_{\mathcal{F}, 1} < +\infty\} \\
	&  \{\boldsymbol{X}_{\mathcal{K}, 1}: D_1(\boldsymbol{X}_{\mathcal{K}, 1}) =\{1\}\} = \{\boldsymbol{X}_{\mathcal{K},1}: \Delta_* < X_{1, 1} \leq (\Delta_* - p_2 X_{2, 1})/p_1\} \\
	&  \{\boldsymbol{X}_{\mathcal{K}, 1}: D_1(\boldsymbol{X}_{\mathcal{K}, 1}) =\{2\}\} = \{\boldsymbol{X}_{\mathcal{K},1}: \Delta_* < X_{2, 1} \leq (\Delta_* - p_1 X_{1, 1})/p_2\}
\end{align*}

\subsection{Subpopulation selection rules for multiple subpopulations case}

\subsubsection{A selection rule under treatment effect monotonicity}
Here, we consider the subpopulation selection rule $D_2$ proposed by \citet{kimani2018point} for adaptive enrichment designs with multiple subpopulations $(\mathcal{K} = \{1,\dots,k\} \; (k\geq 2))$. They assume the true treatment effects in $\mathcal{S}_1$ to $\mathcal{S}_k$ is monotonically decreasing such that $\Delta_{1} \geq \cdots \geq \Delta_{k}$. Let $[m] = \{1,\dots, m\} \; (1 \leq m \leq k)$ denote the set containing the first $m$ positive integers, and ${[m]}^\prime$ denote the complement of $[m]$ such that ${[m]}^\prime = \mathcal{K} \setminus [m]$. Under $D_2$, the largest combined subpopulation $\mathcal{S}_{[m]}$ whose stage 1 sample mean difference is not less than the threshold $\Delta_*$ is selected to continue in stage 2. If $\mathcal{S}_{[m]} < \Delta_*$ for all $m \; (1 \leq m \leq k)$, the trial stops for futility in stage 2.

The selection can be formally characterized by
\begin{equation*}
	\{\boldsymbol{X}_{\mathcal{K}, 1}: D_2(\boldsymbol{X}_{\mathcal{K}, 1}) = [m]\} =\{\boldsymbol{X}_{\mathcal{K}, 1}: X_{[m], 1} \geq \Delta_*, \; X_{[\ell], 1} < \Delta_* \ (\text{for any } m < \ell \leq k) \}.
\end{equation*}
\citet{kimani2018point} demonstrated that the above selection can be re-expressed as
\begin{align*}
	\{\boldsymbol{X}_{\mathcal{K}, 1}: D_2(\boldsymbol{X}_{\mathcal{K}, 1}) = [m]\} = \{\boldsymbol{X}_{\mathcal{K}, 1}: \Delta_* \leq X_{[m], 1} < U_{D_2}([m], \boldsymbol{X}_{{[m]}^\prime, 1}) \},
\end{align*}
where $U_{D_2}([m], \boldsymbol{X}_{{[m]}^\prime, 1}) = \min \{ \frac{p_{[m+1]}\Delta_* - p_{m+1} X_{m+1, 1}}{p_{[m]}}, \frac{p_{[m+2]}\Delta_* - \sum_{i = m+1}^{m+2} p_{i} X_{i, 1}}{p_{[m]}}, \cdots, \frac{p_{[k]}\Delta_* - \sum_{i = m+1}^{k} p_{i} X_{i, 1}}{p_{[m]}} \}$. Accordingly, the stage 1 sample space partition induced by $D_2$ satisfies \cref{class:def} and therefore $D_2$ lies in class $\mathcal{D}$.

\subsubsection{Pick-the-winner selection rule}
Pick-the-winner selection is typically used in adaptive designs with treatment selection and dose-finding \citep{sampson2005drop,bowden2008unbiased,robertson2016accounting}. It may be less appropriate for adaptive enrichment designs, where the broadest possible beneficial subpopulation is often intended to be preserved. Nevertheless, we consider the pick-the-winner selection rule, denoted by $D_3$, as a simple illustrative example. Using $D_3$, the individual subpopulation $\mathcal{S}_w \; (w \in \mathcal{K})$ with the largest stage 1 sample difference is selected to continue in stage 2.

The selection can be formally characterized by
\begin{equation*}
	\{\boldsymbol{X}_{\mathcal{K}, 1}: D_3(\boldsymbol{X}_{\mathcal{K}, 1}) = w \} =\{\boldsymbol{X}_{\mathcal{K}, 1}: X_{w,1} = \max\nolimits_{m \in \mathcal{K}} X_{m,1} \}.
\end{equation*}
Let $L_{D_3}(w,  \boldsymbol{X}_{\{w\}^\prime, 1}) = \max\nolimits_{m \in \{w\}^\prime } X_{m,1}$ and $U_{D_3}(w,  \boldsymbol{X}_{\{w\}^\prime, 1}) = +\infty$, the above selection can be re-expressed as
\begin{equation*}
	\{\boldsymbol{X}_{\mathcal{K}, 1}: D_3(\boldsymbol{X}_{\mathcal{K}, 1}) = w \} =\{\boldsymbol{X}_{\mathcal{K}, 1}: L_{D_3}(w,  \boldsymbol{X}_{\{w\}^\prime, 1}) \leq X_{w,1} < U_{D_3}(w,  \boldsymbol{X}_{\{w\}^\prime, 1}) \},
\end{equation*}
which satisfies \cref{class:def} and thus $D_3$ lies in class $\mathcal{D}$.

\section{Methodology}\label{sec:umvcue}

\subsection{Uniformly minimum variance conditional unbiased estimator}\label{umvcue:derivation}
For a two-stage adaptive enrichment design with patient enrolment follows the scheme described in \cref{sec:notation} and a subpopulation selection rule $D \in \mathcal{D}$ as defined in \cref{sec:class}, we provide a unified expression of the uniformly minimum variance conditional unbiased estimator (UMVCUE) for the treatment effect in the selected subpopulation by \cref{theorem}.

\begin{theorem} \label{theorem}
	For a two-stage adaptive enrichment design with subpopulation selection rule $D \in \mathcal{D}$, the uniformly minimum variance conditional unbiased estimator \pr{UMVCUE} of $\Delta_{\mathcal{E}}$ given $D(\boldsymbol{X}_{\mathcal{K}, 1}) =\mathcal{E}$ is
	\begin{equation}
		X_{\mathcal{E}} +  v_\mathcal{E}  \sqrt{r_{\mathcal{E}}}  \frac{\phi[\frac{\sqrt{r_{\mathcal{E}}}}{v_{\mathcal{E}}} (U_D(\mathcal{E}, \boldsymbol{X}_{\mathcal{E}^\prime, 1}) - X_{\mathcal{E}} )] - \phi[\frac{\sqrt{r_{\mathcal{E}}}}{v_{\mathcal{E}}} (L_D(\mathcal{E}, \boldsymbol{X}_{\mathcal{E}^\prime, 1}) - X_{\mathcal{E}} )]}{\Phi[\frac{\sqrt{r_{\mathcal{E}}}}{v_{\mathcal{E}}} (U_D(\mathcal{E}, \boldsymbol{X}_{\mathcal{E}^\prime, 1}) - X_{\mathcal{E}} )] - \Phi[\frac{\sqrt{r_{\mathcal{E}}}}{v_{\mathcal{E}}}(L_D(\mathcal{E}, \boldsymbol{X}_{\mathcal{E}^\prime, 1}) - X_{\mathcal{E}} )]}.
		\label{UMVCUE:expression}
	\end{equation}
\end{theorem}

\begin{proof}
	The derivation of the UMVCUE is technically involved. We therefore provide a roadmap in the main text, highlighting the key steps, and defer the full derivation to \cref{proof}.
	\begin{enumerate}
		\item We derive the conditional distribution of $(X_{\mathcal{E},1}, X_{\mathcal{E},2}, \boldsymbol{X}_{\mathcal{E}^\prime, 1})$ given $D(\boldsymbol{X}_{\mathcal{K}, 1}) =\mathcal{E}$, which is
		\begin{equation}
			\begin{aligned}
				& f(x_{\mathcal{E}, 1}, x_{\mathcal{E}, 2}, \boldsymbol{x}_{\mathcal{E}^\prime, 1} \mid D(\boldsymbol{X}_{\mathcal{K}, 1}) =\mathcal{E}) \\
				= & \phi([\frac{n_1 p_{\mathcal{E}}/(n_1 p_{\mathcal{E}} + n_2) x_{\mathcal{E}, 1}   + n_2/(n_1 p_{\mathcal{E}} + n_2)x_{\mathcal{E}, 2} ]  - \Delta_{\mathcal{E}}}{v_{\mathcal{E}, 1} v_{\mathcal{E}, 2}/ \sqrt{(v_{\mathcal{E}, 1})^2 + (v_{\mathcal{E}, 2})^2}})\\
				\times & \phi(\frac{x_{\mathcal{E}, 1}  - [n_1 p_{\mathcal{E}}/(n_1 p_{\mathcal{E}} + n_2) x_{\mathcal{E}, 1}   + n_2/(n_1 p_{\mathcal{E}} + n_2)x_{\mathcal{E}, 2}  ]}{(v_{\mathcal{E}, 1})^2 / \sqrt{(v_{\mathcal{E}, 1})^2 + (v_{\mathcal{E}, 2})^2}}) \\
				\times &   \frac{1}{v_{\mathcal{E}, 1}}  \frac{1}{v_{\mathcal{E}, 2}} \prod_{i \in \mathcal{E}^\prime} \bigl[ \frac{1}{v_{i,1}} \phi(\frac{x_{i, 1} - \Delta_i}{v_{i,1}}) \bigr]  \frac{\mathds{1}(L_D(\mathcal{E}, \boldsymbol{X}_{\mathcal{E}^\prime, 1})<x_{\mathcal{E}, 1}<U_D(\mathcal{E}, \boldsymbol{X}_{\mathcal{E}^\prime, 1}))}{\Pr(L_D(\mathcal{E}, \boldsymbol{X}_{\mathcal{E}^\prime, 1})<X_{\mathcal{E}, 1}<U_D(\mathcal{E}, \boldsymbol{X}_{\mathcal{E}^\prime, 1}))}
			\end{aligned}
			\label{eq:cpdf}
		\end{equation}
		\item We show that \cref{eq:cpdf} belongs to the exponential family and is full rank. Hence, $(X_\mathcal{E}, \boldsymbol{X}_{\mathcal{E}^\prime, 1})$ is complete sufficient statistic for $(\Delta_{\mathcal{E}}, \boldsymbol{\Delta}_{\mathcal{E}^\prime})$, where $\boldsymbol{\Delta}_{\mathcal{E}^\prime} = (\Delta_m)_{m \in \mathcal{E}^\prime}$.
		\item By Rao-Blackwell Theorem \citep{rao1945information, blackwell1947conditional} and Lehmann–Scheffé Theorem \citep{lehmann2011completeness}, $\E[X_{\mathcal{E},2} \mid X_\mathcal{E}, \boldsymbol{X}_{\mathcal{E}^\prime, 1}, D(\boldsymbol{X}_{\mathcal{K}, 1}) =\mathcal{E}]$ is the unique UMVCUE of $\Delta_{\mathcal{E}}$. We derive this expression, which is given by
		\begin{align*}
			& \E[X_{\mathcal{E},2} \mid X_\mathcal{E}, \boldsymbol{X}_{\mathcal{E}^\prime, 1}, D(\boldsymbol{X}_{\mathcal{K}, 1}) =\mathcal{E}] \\ = & X_{\mathcal{E}} +  v_\mathcal{E}  \sqrt{r_{\mathcal{E}}}  \frac{\phi[\frac{\sqrt{r_{\mathcal{E}}}}{v_{\mathcal{E}}} (U_D(\mathcal{E}, \boldsymbol{X}_{\mathcal{E}^\prime, 1}) - X_{\mathcal{E}} )] - \phi[\frac{\sqrt{r_{\mathcal{E}}}}{v_{\mathcal{E}}} (L_D(\mathcal{E}, \boldsymbol{X}_{\mathcal{E}^\prime, 1}) - X_{\mathcal{E}} )]}{\Phi[\frac{\sqrt{r_{\mathcal{E}}}}{v_{\mathcal{E}}} (U_D(\mathcal{E}, \boldsymbol{X}_{\mathcal{E}^\prime, 1}) - X_{\mathcal{E}} )] - \Phi[\frac{\sqrt{r_{\mathcal{E}}}}{v_{\mathcal{E}}}(L_D(\mathcal{E}, \boldsymbol{X}_{\mathcal{E}^\prime, 1}) - X_{\mathcal{E}} )]}
		\end{align*}
	\end{enumerate}
\end{proof}

For two-stage adaptive enrichment designs whose subpopulation selection rules do not strictly conform to class $\mathcal{D}$, the sample space partition perspective proposed in \cref{class:def} may remain applicable. We formalize the applicability condition and provide a general expression of the UMVCUE for the treatment effect in an arbitrary subpopulation within the selected population by \cref{corollary}.

\begin{theorem}[Generalized version of \cref{theorem}] \label{corollary}
	For a two-stage adaptive enrichment design with an arbitrary subpopulation selection rule $D^*$. Given the subpopulation selection event $D^*(\boldsymbol{X}_{\mathcal{K}, 1}) = \mathcal{E}$, suppose that for a subpopulation $\mathcal{S}_\mathcal{I} \; (\mathcal{I} \subseteq \mathcal{E})$, the selection event $D^*(\boldsymbol{X}_{\mathcal{K}, 1}) = \mathcal{E}$ with fixed observed value $\boldsymbol{X}_{\mathcal{I}^\prime, 1} = \boldsymbol{x}_{\mathcal{I}^\prime, 1}$ implies its stage 1 sample mean difference is bounded such that
	\begin{gather*}
		\{L_{D^*}(\mathcal{E}, \boldsymbol{x}_{\mathcal{I}^\prime, 1}) < X_{\mathcal{I}, 1} < U_{D^*}(\mathcal{E}, \boldsymbol{x}_{\mathcal{I}^\prime, 1})\},
	\end{gather*}
	Then, the uniformly minimum variance conditional unbiased estimator \pr{UMVCUE} of $\Delta_{\mathcal{I}}$ given the selection $D^*(\boldsymbol{X}_{\mathcal{K}, 1}) =\mathcal{E}$ is
	\begin{equation}
		X_{\mathcal{I}} +  v_\mathcal{I}  \sqrt{r_{\mathcal{I}}}  \frac{\phi[\frac{\sqrt{r_{\mathcal{I}}}}{v_{\mathcal{I}}} (U_D(\mathcal{E}, \boldsymbol{X}_{\mathcal{I}^\prime, 1}) - X_{\mathcal{I}} )] - \phi[\frac{\sqrt{r_{\mathcal{I}}}}{v_{\mathcal{I}}} (L_D(\mathcal{E}, \boldsymbol{X}_{\mathcal{I}^\prime, 1}) - X_{\mathcal{I}} )]}{\Phi[\frac{\sqrt{r_{\mathcal{I}}}}{v_{\mathcal{I}}} (U_D(\mathcal{E}, \boldsymbol{X}_{\mathcal{I}^\prime, 1}) - X_{\mathcal{I}} )] - \Phi[\frac{\sqrt{r_{\mathcal{I}}}}{v_{\mathcal{I}}} (L_D(\mathcal{E}, \boldsymbol{X}_{\mathcal{I}^\prime, 1}) - X_{\mathcal{I}} )]}.
		\label{UMVCUE:expression2}
	\end{equation}
\end{theorem}

\begin{proof}
	The proof is completely analogous with that of \cref{theorem} in \cref{proof}, where $(X_{\mathcal{I}, 1}, \boldsymbol{X}_{\mathcal{I}^\prime, 1})$ is the complete sufficient statistic for $(\Delta_{\mathcal{I}}, \boldsymbol{\Delta}_{\mathcal{I}^\prime})$ where $\boldsymbol{\Delta}_{\mathcal{I}^\prime} = (\Delta_m)_{m \in \mathcal{I}^\prime}$, and $X_{\mathcal{I}, 2}$ is an unbiased estimator of $\Delta_{\mathcal{I}}$, given that $D^*(\boldsymbol{X}_{\mathcal{K}, 1}) = \mathcal{E}$.
\end{proof}

When the target subpopulation $\mathcal{S}_{\mathcal I}$ in \cref{corollary} is an individual subpopulation, so that $|\mathcal I| = 1$, the UMVCUE of its treatment effect is given in the following corollary.

\begin{corollary}[Corollary to \cref{corollary}] \label{corollary2}
	For a two-stage adaptive enrichment design with an arbitrary subpopulation selection rule $D^*$. Given the subpopulation selection event $D^*(\boldsymbol{X}_{\mathcal{K}, 1}) = \mathcal{E}$, suppose that for an individual subpopulation $\mathcal{S}_m \; (m \in \mathcal{E})$, the selection event $D^*(\boldsymbol{X}_{\mathcal{K}, 1}) = \mathcal{E}$ with fixed observed value $\boldsymbol{X}_{\{m\}^\prime, 1} = \boldsymbol{x}_{\{m\}^\prime, 1}$ implies its stage 1 sample mean difference is bounded such that
	\begin{gather*}
		\{L_{D^*}(\mathcal{E}, \boldsymbol{x}_{\{m\}^\prime, 1}) < X_{m, 1} < U_{D^*}(\mathcal{E}, \boldsymbol{x}_{\{m\}^\prime, 1})\},
	\end{gather*}
	Then, the uniformly minimum variance conditional unbiased estimator \pr{UMVCUE} of $\Delta_{m}$ given the selection $D^*(\boldsymbol{X}_{\mathcal{K}, 1}) =\mathcal{E}$ is
	\begin{equation}
		X_{m} +  v_m  \sqrt{r_{m}}  \frac{\phi[\frac{\sqrt{r_{m}}}{v_{m}} (U_D(\mathcal{E}, \boldsymbol{X}_{\{m\}^\prime, 1}) - X_{m} )] - \phi[\frac{\sqrt{r_{m}}}{v_{m}} (L_D(\mathcal{E}, \boldsymbol{X}_{\{m\}^\prime, 1}) - X_{m} )]}{\Phi[\frac{\sqrt{r_{m}}}{v_{m}} (U_D(\mathcal{E}, \boldsymbol{X}_{\{m\}^\prime, 1}) - X_{m} )] - \Phi[\frac{\sqrt{r_{m}}}{v_{m}} (L_D(\mathcal{E}, \boldsymbol{X}_{\{m\}^\prime, 1}) - X_{m} )]}.
		\label{UMVCUE:expression3}
	\end{equation}
\end{corollary}

\subsection{Conditional estimator when the outcome variance $\sigma^2$ is unknown}

So far in the paper, we assume the outcome variance $\sigma^2$ is known. Note that, the expressions of UMVCUEs in \cref{UMVCUE:expression} and \cref{UMVCUE:expression2} involves $\sigma$ implicitly by $v_\mathcal{E}$ and $v_\mathcal{I}$, respectively. For the unknown variance case, \citet{robertson2019conditionally} derived the UMVCUE for adaptive designs with treatment selection, where the best performing treatment is selected to continue in stage 2. \citet{robertson2019conditionally} compared the performance, including bias, mean squared error, and distribution, of the UMVCUE derived with unknown variance and the UMVCUE derived with known variance and plug-in estimated variance. The simulation showed that the performance of these two forms of UMVCUEs are comparable in the vast majority of cases. Hence, for the sake of tractability and practical utility, we provide a plug-in conditional estimator for adaptive enrichment designs when the variance is unknown by \cref{PiCE:def}.

\begin{definition}[Plug-in Conditional Estimator] \label{PiCE:def}
	When the variance $\sigma^2$ is unknown, the Plug-in Conditional Estimator (PiCE) is defined by substituting $\sigma^2$ with its estimator $\hat{\sigma}^2$ in $v_\mathcal{E}$ and $v_\mathcal{I}$ within the expression of the UMVCUE in \cref{UMVCUE:expression} and \cref{UMVCUE:expression2}, respectively. $\hat{\sigma}^2$ is given by
	\begin{gather}
		\hat{\sigma}^2 = \frac{\sum_{m \in \mathcal{E}, a \in \{0,1\}} SSW_{ma} }{n_1+n_2 - 2k}, \label{PiCE} \\
		\text{where } \; SSW_{ma} = \sum\nolimits_{\{i:M_i = m , A_i = a\}} (Y_i - \bar{Y}_{ma} )^2, \; \; \bar{Y}_{ma} = \frac{\sum_{\{i:M_i = m , A_i = a\}} Y_i }{|\{i:M_i = m , A_i = a \}|} \notag 
	\end{gather}
\end{definition}

\section{Worked Example}\label{sec:example}

To illustrate the practical use of the proposed UMVCUE framework, we revisit the simplified adaptive enrichment trial considered by \citet{li2026confidence}, which was motivated by the ENRICH trial. The Early Minimally Invasive Removal of Intracerebral Hemorrhage (ENRICH) trial \citep{pradilla2024trial} adopted a multi-stage adaptive enrichment design to account for potential heterogeneity in treatment effects across patient subpopulations. Patients were classified by haemorrhage location into those with lobar haemorrhage ($\mathcal{S}_1$) and those with anterior basal ganglia haemorrhage ($\mathcal{S}_2$). The planned sample size ranged from 150 to 300 patients, with interim analyses scheduled after the enrolment of 150, 175, 200, 225, 250, and 275 patients. Ultimately, 300 patients were recruited. After the interim analysis at 175 patients, enrolment was restricted to the lobar haemorrhage subgroup ($\mathcal{S}_1$). In the confirmatory analysis, treatment effects in $\mathcal{S}_1$, $\mathcal{S}_2$, and $\mathcal{F}$ were estimated. The use of subgroup analysis, multi-stage interim decision-making, and adaptive enrolment makes the trial closely aligned with the adaptive enrichment designs considered in this work.

In the simplified design, we assume a maximum sample size of 300 patients, with a single interim analysis conducted after 200 patients using the subpopulation selection rule $D_1$ defined in \cref{sec:design1}. Consistent with the ENRICH trial, the primary endpoint is the utility-weighted modified Rankin scale at 180 days. Following published studies \citep{mendelow2013early,hanley2016safety,pradilla2024trial}, we assume that the endpoint is normally distributed with a common variance across treatment arms and subpopulations, taking $\sigma^2 = 0.36^2$. For analytical simplicity, the subgroup prevalences are set to $p_1 = p_2 = 0.5$. This approximation is in line with the early accrual observed in the ENRICH trial, where, among the first 175 enrolled patients, 83 (47.4\%) had lobar haemorrhage and 92 (52.6\%) had anterior basal ganglia haemorrhage. The interim decision threshold is set to $\Delta^* = 0.025$, chosen so that, when the true treatment effect in the overall population equals the minimum clinically meaningful effect reported in the ENRICH trial, $\Delta_{\min} = 0.075$, the stage 1 estimator $X_{\mathcal F,1}$ exceeds $\Delta^*$ with probability greater than $80\%$. This simplified design falls within the class of designs to which the proposed estimation framework applies.

Stage-wise summaries from the ENRICH trial are not reported. Therefore, for illustration purposes, \citet{li2026confidence} simulated stage-wise sample means that are numerically consistent with the reported overall sample means in $\mathcal{S}_1$ and $\mathcal{S}_2$. In the present illustration, we use the same values. The stage 1 sample means in $\mathcal{S}_1$, $\mathcal{S}_2$, and the overall population $\mathcal{F}$ are $x_{1, 1} = 0.113$, $x_{2,1} = 0.013$, and $x_{\mathcal{F},1} = 0.063$, respectively. According to the decision rule, the trial proceeds with the overall population, i.e., $D_1(\boldsymbol{X}_{\mathcal{K},1}) = \mathcal{K} = \{1,2\}$. At stage 2, the sample means are $x_{1, 2}  = 0.155$, $x_{2, 2}  = -0.064$, and $x_{\mathcal{F},2} = 0.045$. The corresponding overall naive estimates are $x_{1}= 0.127$, $x_{2} = -0.013$, and $x_{\mathcal{F}} = 0.057$. 

\begin{table}[!htbp]
	\centering
	\begin{tabular}{c c c c}
		\hline
		Method & estimates for $\Delta_{\mathcal{F}}$ & estimates for $\Delta_{1}$ & estimates for $\Delta_{2}$ \\ \hline
		MLE & 0.057& 0.127 & -0.013 \\
		UMVCUE & 0.042 & 0.124 & -0.031 \\ \hline	
	\end{tabular}
	\caption{MLEs and UMVCUEs for treatment effects in different patient populations.}
	\label{tab:example}
\end{table}

Using these quantities, we calculate the UMVCUEs for the treatment effects in $\mathcal{F}$, $\mathcal{S}_1$, and $\mathcal{S}_2$ as follows. Since the trial continues with the overall population, from \cref{sec:design1}, we have showed the induced stage 1 sample space partition satisfying $\{\boldsymbol{X}_{\mathcal{K},1}: \Delta_* < X_{\mathcal{F}, 1} < +\infty\}$. Then, we can apply  \cref{theorem} to calculate the UMVCUE of $\Delta_\mathcal{F}$ by setting $L_{D_1}(\mathcal{K}, \boldsymbol{X}_{\mathcal{K}^\prime, 1}) = \Delta_*$, $U_{D_1}(\mathcal{K}, \boldsymbol{X}_{\mathcal{K}^\prime, 1}) = +\infty$, $x_{\mathcal{F}} = 0.057$, and other fixed quantities in \cref{UMVCUE:expression}. To obtain estimates of $\Delta_1$ and $\Delta_2$, note that, given $D_1(\boldsymbol{X}_{\mathcal{K},1}) = \mathcal{K}$, the stage 1 sample space partition can be re-expressed as $\{\boldsymbol{X}_{\mathcal{K},1}: (\Delta_* - p_2 X_{2, 1})/p_1 < X_{1, 1} < +\infty\}$. In this case, \cref{corollary2} is applicable to obtain the UMVCUE for $\Delta_1$, where $X_{1, 1}$ served as a complete and sufficient statistic. Then, by setting $L_{D_1}(\mathcal{K}, \boldsymbol{X}_{{\{1\}}^\prime, 1}) =(\Delta_* - p_2 x_{2, 1})/p_1 $, $U_{D_1}(\mathcal{K}, \boldsymbol{X}_{{\{1\}}^\prime, 1}) = +\infty$, $x_{1} = 0.127$, and other fixed quantities into \cref{UMVCUE:expression2}, we can calculate the UMVCUE of $\Delta_{1}$. A completely analogous procedure can be applied to calculate the UMVCUE for $\Delta_2$. The values of MLEs and UMVCUEs are compared in \cref{tab:example}.

Theoretically, the conventional MLE ignores the data-dependent subpopulation selection and treats the data as arising from a fixed design, which is often biased. In contrast, the derivation of UMVCUEs is conditional on the adaptive selection and therefore adjusts for the adaptive nature of the design, resulting in conditional unbiased estimation. In this example design, the trial continues with the overall population in stage 2 when the observed stage 1 data are promising, which is also indicated by the induced sample space partition, such that there are lower thresholds for $X_{\mathcal{F}, 1}$, $X_{1, 1}$, and $X_{2, 1}$. Hence, in this case, MLEs are more likely to overestimate the treatment effects; while, the value of UMVCUEs are consistently less than those of MLEs across the overall population and both subpopulations, which suggests UMVCUEs are able to correct such overestimation.

\section{Simulation Study} \label{sec:sim}

We consider the case where the overall patient population $\mathcal{F}$ can be partitioned into two subpopulations ($\mathcal{S}_1$ and $\mathcal{S}_2$), and the adaptive enrichment design with subpopulation selection rule $D_1$ described in \cref{sec:design1} is used. Suppose patient enrolment follows the scheme described in \cref{sec:notation} and the proportion of each subpopulation is $p_1 = p_2 = 0.5$. Five hypothetical scenarios that represent different degrees of treatment effect and heterogeneity are considered. In scenario 1, we assume the experimental treatment benefits the overall population equally such that $\Delta_1 = \Delta_2 = 0.5$. In scenario 2, suppose $\mathcal{S}_1$ benefits more from the treatment than $\mathcal{S}_2$ such that $\Delta_1 = 0.5$ and $\Delta_2 = 0.2$. In scenario 3, we assume the treatment only benefits $\mathcal{S}_1$ such that $\Delta_1 = 0.5$ and $\Delta_2 = 0$. In scenario 4, there is no benefit in either subpopulation, such that $\Delta_1 = \Delta_2 = 0$. In scenario 5, suppose the treatment benefits $\mathcal{S}_1$ but has side effects in $\mathcal{S}_2$, such that $\Delta_1 = 0.5$ and $\Delta_2 = -0.2$. In each scenario, the outcome means in treatment and control arms in $\mathcal{S}_m \; (m \in \{1,2\})$ are set to be $\mu_{m1} = \Delta_m$ and $\mu_{m1} = 0$, respectively. We set $\sigma^2 = 1$ and stage-wise sample size to be $n_1 = n_2 = 100$. In this setting, a standard fixed design with Type I error $\alpha = 0.05$ can achieve over $90\%$ power to reject the null hypothesis $\Delta_\mathcal{F} \leq 0$ under scenario 1. The threshold in the subpopulation selection rule is set to be $\Delta_* = 0.3$, which yields reasonable population selections across the five scenarios in the simulation. Under each scenario, we conducted 100,000 simulated trials. For each simulated trial, we selected the target population based on stage 1 data using rule $D_1$ and calculated the MLE, UMVCUE, and PiCE using \cref{MLE}, \cref{UMVCUE:expression}, and \cref{PiCE}, respectively. Note that, in the calculation of PiCE, only $\hat{\sigma}^2$ is used, to simulate the case when the variance is unknown. 

\begin{table}[!htbp] 
	\centering
	\begin{tabular}{l c c c c c}
		\hline
		Scenario & Method & Continue with $\mathcal{F}$ & Enrich to $\mathcal{S}_1 $ & Enrich to $\mathcal{S}_2 $ & Stop \\ \hline
		\multirow{4}{*}{\makecell[l]{Scenario 1: \\ $\Delta_1 = 0.5$  \\ $\Delta_2 = 0.5$ }} & Selection proportion & $84.17\%$ & $5.06\%$ & $5.03\%$ & $5.74\%$ \\
		& MLE & $28.71 \ (17.07)$ & $-17.36 \ (19.22)$ & $-22.33 \ (19.57) $ & $-$   \\
		& UMVCUE & $-0.05 \ (24.76)$ & $4.02 \ (37.43)$ & $-2.36 \ (38.08)$ & $-$   \\
		& PiCE & $-0.03 \ (24.75)$ & $4.01 \ (37.41)$ & $-2.34 \ (38.06)$ & $-$   \\ \hline
		\multirow{4}{*}{\makecell[l]{Scenario 2: \\ $\Delta_1 = 0.5$  \\ $\Delta_2 = 0.2$ }} & Selection proportion & $59.64\%$ & $21.13\%$ & $3.97\%$ & $15.27\%$ \\
		& MLE & $64.45  \ (18.33)$ & $-6.81 \ (19.81)$ & $66.83 \ (23.03) $ & $-$   \\
		& UMVCUE & $-0.15 \ (28.29)$ & $ -0.98 \ (37.26)$ & $-0.21 \ (38.70)$ & $-$   \\
		& PiCE  & $-0.05 \ (28.31)$ & $-0.95 \ (37.22)$ & $-0.15 \ (38.71)$ & $-$   \\ \hline
		\multirow{4}{*}{\makecell[l]{Scenario 3: \\ $\Delta_1 = 0.5$  \\ $\Delta_2 = 0$ }} & Selection proportion & $40.21\%$ & $37.45\%$ & $1.83\%$ & $20.51\%$ \\ 
		& MLE & $97.01 \ (22.43)$ & $5.02 \ (20.29)$ & $127.25 \ (34.65) $ & $-$   \\
		& UMVCUE & $1.09 \ (30.43)$ & $-0.50\ (36.04)$ & $-2.10 \ (38.73)$ & $-$   \\
		& PiCE & $1.19 \ (30.50)$ & $-0.48 \ (36.02)$ & $-2.12 \ (38.79)$ & $-$   \\ \hline
		\multirow{4}{*}{\makecell[l]{Scenario 4: \\ $\Delta_1 = 0$  \\ $\Delta_2 = 0$ }} & Selection proportion & $6.76 \%$ & $9.92\%$ & $10.14\%$ & $73.18\%$ \\
		& MLE & $195.07 \ (49.88)$ & $139.30 \ (38.36)$ & $138.89 \ (38.52) $ & $-$   \\
		& UMVCUE & $1.93 \ (35.88)$ & $ -0.14 \ (37.96)$ & $ -0.87 \ (38.36)$ & $-$   \\
		& PiCE & $ 1.58 \ (36.03)$ & $-0.25 \ (38.03)$ & $ -0.89 \ (38.42)$ & $-$   \\ \hline
		\multirow{4}{*}{\makecell[l]{Scenario 5: \\ $\Delta_1 = 0.5$  \\ $\Delta_2 = -0.2$ }} & Selection proportion & $22.88 \%$ & $53.62\%$ & $0.57\%$ & $22.93\%$ \\
		& MLE & $133.16 \ (30.09)$ & $16.90 \ (21.24)$ & $187.93 \ (53.43) $ & $-$   \\
		& UMVCUE & $0.45 \ (32.78)$ & $0.50 \ (34.93)$ & $-4.47 \ (38.76)$ & $-$   \\
		& PiCE & $0.40 \ (32.88)$ & $0.55 \ (34.98)$ & $-4.54 \ (38.80)$ & $-$   \\ \hline
	\end{tabular}
	\caption{Simulated subpopulation selection proportions, bias $\times 10^{3}$, and mean squared error $\times 10^{3}$ across five scenarios.}
	\label{tab:sim}
\end{table}

The simulation results are summarised in \cref{tab:sim}. In each scenario block, the first row indicates the proportion of subpopulation selections at the interim analysis. The following rows summarie the conditional bias and mean squared error (MSE) of MLE, UMVCUE, and PiCE given the selected population, where values in parentheses denote the MSE. Note that bias and MSE values are scaled by $10^3$ for clarity. Based on the subpopulation selection proportions, $D_1$ yields reasonable decisions. In scenario 1, $84.17\%$ of simulated trials continue with the overall population $\mathcal{F}$. In scenario 2, an appreciable fraction $(21.13\%)$ of trials are enriched to $\mathcal{S}_1$, though the majority of trials continue with $\mathcal{F}$ since the aggregated treatment effect $\Delta_F = 0.35$ is still considerable. In scenario 3, a greater fraction $(21.13\%)$ of trials are enriched to $\mathcal{S}_1$ with the increase of heterogeneity between $\mathcal{S}_1$ and $\mathcal{S}_2$. In scenario 4, where neither subpopulation benefit, $73.18\%$ of trials stop for futility at the interim. In scenario 5, when there is a side effect in $\mathcal{S}_2$, the majority of trials are enriched to $\mathcal{S}_1$.

In all scenarios, the bias of the UMVCUE is negligible, which confirms its conditional unbiasedness. In contrast, the MLE exhibits substantial bias, particularly when a mismatched subpopulation is selected, which is due to the selection bias induced by the data-dependent selection. For example, in scenario 4, the trials are supposed to stop for futility since there is no benefit in either $\mathcal{S}_1$ or $\mathcal{S}_2$. In this case, a trial continuing with the mismatched subpopulation indicates that the corresponding stage 1 sample mean difference in that subpopulation happens to be great and thus pass the selection, leading to a considerable overestimation of the treatment effect. While the UMVCUE entails a higher variance, leading to greater MSE values in most instances, the bias of the MLE can dominate (e.g. all population selections in scenario 4 and the enrichment to $\mathcal{S}_2$ in scenario 5), resulting a greater MSE than the UMVCUE. Across all scenarios, the PiCE shows similar empirical performance with the UMVCUE regarding bias and MSE. This suggests that, when the variance is unknown, the PiCE is a viable alternative of the UMVCUE. These results highlight that the UMVCUE and PiCE provide a more reliable estimation in post selection inference, especially in ensuring that treatment benefits are not overestimated.

\section{Discussion} \label{sec:discuss}
In this work, we have identified a critical research gap that the majority of point estimation methods for adaptive enrichment designs are tailored to specific subpopulation selection rules, limiting their general applicability to adaptive enrichment designs in practice. To bridge this gap, we propose a general class of subpopulation selection rules that encompasses the majority of those rules described in the current literature. For any adaptive enrichment design with selection rule in this class, we established a systematic framework that derives the Uniformly Minimum Variance Conditional Unbiased Estimator (UMVCUE). Our simulation confirms that the proposed UMVCUE is unbiased and effectively avoids overestimation on treatment effects compared to MLE. Based on the simulation results, when the outcome variance is unknown, the proposed plug-in conditional estimator provides a pragmatic alternative to the UMVCUE. By providing a unified expression of the UMVCUE, our work eliminates the necessity for case-specific derivations of point estimators for each new design. The proposed framework has the potential to facilitate the implementation of diverse subpopulation selection rules with greater ease in real-world adaptive trials and ultimately promote more efficient and ethical drug development.

We emphasize that the class $\mathcal{D}$ should be viewed as a structural characterization or a set of compatibility constraints on selection rules, rather than a prescriptive selection rule itself. Specifically, $\mathcal{D}$ provides a formal boundary to verify the eligibility of a given selection rule $D$ for our conditional point estimation framework. We do not intend for clinicians or biostatisticians to contrive their decision-making processes to satisfy the requirements of $\mathcal{D}$. Instead, we advocate that selection rules remain motivated by scientific and clinical considerations. Whenever such a scientifically-driven rule falls within $\mathcal{D}$, our framework can be applied to ensure rigorous post-selection inference.

Our current methodology is derived under normal distributional assumptions. For non-normal outcomes, the theoretical properties established above hold asymptotically, provided that, under standard regularity conditions, the treatment effect MLE is asymptotically normal as the stage-wise sample sizes tend to infinity. Beyond the two-stage case discussed here, the generalization to multi-stage configurations represents a natural direction for future research, but remains technically challenging and is left as an open problem.

In this paper, we have focused on conditional inference. The merits and drawbacks between conditional and unconditional inference lies beyond the scope of this paper. We refer readers to \citet{ohman2003conditional, marschner2021general} for a comprehensive discussion.

\section*{Acknowledgement}
This work was supported by the MRC Doctoral Training Partnership in Interdisciplinary Biomedical Research awarded to the University of Warwick (MR/W007053/1) and by Novartis. The authors thank Dominic Magirr for valuable discussions related to the ideas underlying this work.

\bibliographystyle{apalike}
\bibliography{ref}

\appendix
\crefalias{section}{appendix}
\crefalias{subsection}{appendix}

\crefname{appendix}{Appendix}{Appendices}
\Crefname{appendix}{Appendix}{Appendices}

\newpage
\section{Appendix}

\subsection{Proof of \cref{theorem}}  \label{proof}

\begin{proof}
	In the following derivations, where the design $D$ and the selection result $\mathcal{E}$ are fixed, we suppress the explicit dependence for brevity and write the induced boundaries, $L_D(\mathcal{E}, \boldsymbol{X}_{\mathcal{E}^\prime, 1})$ and $U_D(\mathcal{E}, \boldsymbol{X}_{\mathcal{E}^\prime, 1})$,  simply as $L$ and $U$.
	
	Consider the joint conditional distribution of $(X_{\mathcal{E},1}, X_{\mathcal{E},2}, \boldsymbol{X}_{\mathcal{E}^\prime, 1})$ given $D(\boldsymbol{X}_{\mathcal{K}, 1}) =\mathcal{E}$, which is
	\begin{align*}
		& f(x_{\mathcal{E}, 1}, x_{\mathcal{E}, 2}, \boldsymbol{x}_{\mathcal{E}^\prime, 1} \mid D(\boldsymbol{X}_{\mathcal{K}, 1}) =\mathcal{E}) \\
		= &  \frac{1}{v_{\mathcal{E}, 1}}\phi (\frac{x_{\mathcal{E}, 1} - \Delta_{\mathcal{E}}}{v_{\mathcal{E}, 1}})   \frac{1}{v_{\mathcal{E}, 2}} \phi (\frac{x_{\mathcal{E}, 2} - \Delta_{\mathcal{E}}}{v_{\mathcal{E}, 2}}) \prod_{i \in \mathcal{E}^\prime} \bigl[ \frac{1}{v_{i,1}} \phi(\frac{x_{i, 1} - \Delta_i}{v_{i,1}}) \bigr]  \frac{\mathds{1}(L<x_{\mathcal{E}, 1}<U)}{\Pr(L<X_{\mathcal{E}, 1}<U)}
	\end{align*}
	Note that, the following part of the expression above can be written as
	\begin{align*}
		& \phi (\frac{x_{\mathcal{E}, 1} - \Delta_{\mathcal{E}}}{v_{\mathcal{E}, 1}}) \phi (\frac{x_{\mathcal{E}, 2} - \Delta_{\mathcal{E}}}{v_{\mathcal{E}, 2}}) \\
		= & \phi([\frac{n_1 p_{\mathcal{E}}/(n_1 p_{\mathcal{E}} + n_2) x_{\mathcal{E}, 1}   + n_2/(n_1 p_{\mathcal{E}} + n_2)x_{\mathcal{E}, 2} ]  - \Delta_{\mathcal{E}}}{v_{\mathcal{E}, 1} v_{\mathcal{E}, 2}/ \sqrt{(v_{\mathcal{E}, 1})^2 + (v_{\mathcal{E}, 2})^2}})\\
		\times & \phi(\frac{x_{\mathcal{E}, 1}  - [n_1 p_{\mathcal{E}}/(n_1 p_{\mathcal{E}} + n_2) x_{\mathcal{E}, 1}   + n_2/(n_1 p_{\mathcal{E}} + n_2)x_{\mathcal{E}, 2}  ]}{(v_{\mathcal{E}, 1})^2 / \sqrt{(v_{\mathcal{E}, 1})^2 + (v_{\mathcal{E}, 2})^2}})
	\end{align*}
	Then, we have
	\begin{equation}
		\begin{aligned}
			& f(x_{\mathcal{E}, 1}, x_{\mathcal{E}, 2}, \boldsymbol{x}_{\mathcal{E}^\prime, 1} \mid D(\boldsymbol{X}_{\mathcal{K}, 1}) =\mathcal{E}) \\
			= & \phi([\frac{n_1 p_{\mathcal{E}}/(n_1 p_{\mathcal{E}} + n_2) x_{\mathcal{E}, 1}   + n_2/(n_1 p_{\mathcal{E}} + n_2)x_{\mathcal{E}, 2} ]  - \Delta_{\mathcal{E}}}{v_{\mathcal{E}, 1} v_{\mathcal{E}, 2}/ \sqrt{(v_{\mathcal{E}, 1})^2 + (v_{\mathcal{E}, 2})^2}})\\
			\times & \phi(\frac{x_{\mathcal{E}, 1}  - [n_1 p_{\mathcal{E}}/(n_1 p_{\mathcal{E}} + n_2) x_{\mathcal{E}, 1}   + n_2/(n_1 p_{\mathcal{E}} + n_2)x_{\mathcal{E}, 2}  ]}{(v_{\mathcal{E}, 1})^2 / \sqrt{(v_{\mathcal{E}, 1})^2 + (v_{\mathcal{E}, 2})^2}}) \\
			\times &   \frac{1}{v_{\mathcal{E}, 1}}  \frac{1}{v_{\mathcal{E}, 2}} \prod_{i \in \mathcal{E}^\prime} \bigl[ \frac{1}{v_{i,1}} \phi(\frac{x_{i, 1} - \Delta_i}{v_{i,1}}) \bigr]  \frac{\mathds{1}(L<x_{\mathcal{E}, 1}<U)}{\Pr(L<X_{\mathcal{E}, 1}<U)}
		\end{aligned}
		\label{eq:3}
	\end{equation}
	Note that $f(x_{\mathcal{E}, 1}, x_{\mathcal{E}, 2}, \boldsymbol{x}_{\mathcal{E}^\prime, 1} \mid D(\boldsymbol{X}_{\mathcal{K}, 1}) =\mathcal{E})$ belongs to the exponential family and is full rank with representation \cref{eq:3} (See Section 5, Chapter 1 of \citet{lehmann1998theory}). We have $X_\mathcal{E} = n_1 p_{\mathcal{E}}/(n_1 p_{\mathcal{E}} + n_2) X_{\mathcal{E}, 1}   + n_2/(n_1 p_{\mathcal{E}} + n_2)X_{\mathcal{E}, 2}$. Hence, $(X_\mathcal{E}, \boldsymbol{X}_{\mathcal{E}^\prime, 1})$ is complete sufficient statistic for $(\Delta_{\mathcal{E}}, \boldsymbol{\Delta}_{\mathcal{E}^\prime})$ (See Corollary 6.16 and Theorem 6.22 in Chapter 1 of \citet{lehmann1998theory}). Since $\E[X_{\mathcal{E},2} \mid D(\boldsymbol{X}_{\mathcal{K}, 1}) =\mathcal{E}] = \Delta_{\mathcal{E}}$, by Rao-Blackwell Theorem \citep{rao1945information, blackwell1947conditional} (See Theorem 7.8 in Chapter 1 of \citet{lehmann1998theory}), we have $\E[X_{\mathcal{E},2} \mid X_\mathcal{E}, \boldsymbol{X}_{\mathcal{E}^\prime, 1}, D(\boldsymbol{X}_{\mathcal{K}, 1}) =\mathcal{E}] = \Delta_{\mathcal{E}}$. Based on Lehmann–Scheffé Theorem \citep{lehmann2011completeness} (See Theorem 1.11 in Chapter 2 of \citep{lehmann1998theory}), $\E[X_{\mathcal{E},2} \mid X_\mathcal{E}, \boldsymbol{X}_{\mathcal{E}^\prime, 1}, D(\boldsymbol{X}_{\mathcal{K}, 1}) =\mathcal{E}]$ is the unique UMVCUE of $\Delta_{\mathcal{E}}$. In the following, we derive the expression of $\E[X_{\mathcal{E},2} \mid X_\mathcal{E}, \boldsymbol{X}_{\mathcal{E}^\prime, 1}, D(\boldsymbol{X}_{\mathcal{K}, 1}) =\mathcal{E}]$.
	
	First, we can obtain the conditional distribution of $(X_{\mathcal{E},1}, X_{\mathcal{E}}, \boldsymbol{X}_{\mathcal{E}^\prime, 1})$ given $D(\boldsymbol{X}_{\mathcal{K}, 1}) =\mathcal{E}$ by transforming $(X_{\mathcal{E},1}, X_{\mathcal{E},2}, \boldsymbol{X}_{\mathcal{E}^\prime, 1})$ to $(X_{\mathcal{E},1}, X_{\mathcal{E}}, \boldsymbol{X}_{\mathcal{E}^\prime, 1})$ in \cref{eq:3} as
	\begin{align*}
		& f(x_{\mathcal{E}, 1}, x_{\mathcal{E}}, \boldsymbol{x}_{\mathcal{E}^\prime, 1} \mid D(\boldsymbol{X}_{\mathcal{K}, 1}) =\mathcal{E}) \\
		= & \frac{1}{v_\mathcal{E}} \phi(\frac{x_\mathcal{E} - \Delta_{\mathcal{E}}}{v_\mathcal{E}}) \frac{1}{\sqrt{1/r_\mathcal{E}} v_\mathcal{E} }  \phi(\frac{x_{\mathcal{E}, 1} - x_{\mathcal{E}}}{\sqrt{1/r_\mathcal{E}} v_\mathcal{E}}) \prod_{i \in \mathcal{E}^\prime} \bigl[ \frac{1}{v_{i,1}} \phi(\frac{x_{i, 1} - \Delta_i}{v_{i,1}}) \bigr]  \frac{\mathds{1}(L<x_{\mathcal{E}, 1}<U)}{\Pr(L<X_{\mathcal{E}, 1}<U)}
	\end{align*}
	Then, we can obtain the marginally conditional distribution of $(X_{\mathcal{E}}, \boldsymbol{X}_{\mathcal{E}^\prime, 1})$ given $D(\boldsymbol{X}_{\mathcal{K}, 1}) =\mathcal{E}$ by integrating out $X_{\mathcal{E},1}$ in the above equation as
	\begin{align*}
		& f(x_{\mathcal{E}}, \boldsymbol{x}_{\mathcal{E}^\prime, 1} \mid D(\boldsymbol{X}_{\mathcal{K}, 1}) =\mathcal{E}) \\
		= & \int_{L}^{U} f(x_{\mathcal{E}, 1}, x_{\mathcal{E}}, \boldsymbol{x}_{\mathcal{E}^\prime, 1} \mid D(\boldsymbol{X}_{\mathcal{K}, 1}) =\mathcal{E}) \dif x_{\mathcal{E}, 1} \\
		= &  \frac{1}{v_\mathcal{E}} \phi(\frac{x_\mathcal{E} - \Delta_{\mathcal{E}}}{v_\mathcal{E}}) [  \Phi(\frac{U - x_{\mathcal{E}}}{\sqrt{1/r_\mathcal{E}} v_\mathcal{E}}) - \Phi(\frac{L- x_{\mathcal{E}}}{\sqrt{1/r_\mathcal{E}} v_\mathcal{E}}) ]\prod_{i \in \mathcal{E}^\prime} \bigl[ \frac{1}{v_{i,1}} \phi(\frac{x_{i, 1} - \Delta_i}{v_{i,1}}) \bigr]  \frac{\mathds{1}(L<x_{\mathcal{E}, 1}<U)}{\Pr(L<X_{\mathcal{E}, 1}<U)}
	\end{align*}
	Using an analogous procedure aforementioned, we can transfer $(X_{\mathcal{E},1}, X_{\mathcal{E},2}, \boldsymbol{X}_{\mathcal{E}^\prime, 1})$ to $(X_{\mathcal{E},2}, X_{\mathcal{E}}, \boldsymbol{X}_{\mathcal{E}^\prime, 1})$ in \cref{eq:3} to obtain
	\begin{align*}
		& f(x_{\mathcal{E}, 2}, x_{\mathcal{E}}, \boldsymbol{x}_{\mathcal{E}^\prime, 1} \mid D(\boldsymbol{X}_{\mathcal{K}, 1}) =\mathcal{E}) \\
		= & \frac{1}{v_\mathcal{E}} \phi(\frac{x_\mathcal{E} - \Delta_{\mathcal{E}}}{v_\mathcal{E}}) \frac{1}{\sqrt{r_\mathcal{E}} v_\mathcal{E} }  \phi(\frac{x_{\mathcal{E}, 2} - x_{\mathcal{E}}}{\sqrt{r_\mathcal{E}} v_\mathcal{E}}) \prod_{i \in \mathcal{E}^\prime} \bigl[ \frac{1}{v_{i,1}} \phi(\frac{x_{i, 1} - \Delta_i}{v_{i,1}}) \bigr]  \frac{\mathds{1}(L<x_{\mathcal{E}, 1}<U)}{\Pr(L<X_{\mathcal{E}, 1}<U)}
	\end{align*}
	Then, we can obtain
	\begin{align*}
		& f(x_{\mathcal{E}, 2} \mid X_{\mathcal{E}}, \boldsymbol{X}_{\mathcal{E}^\prime, 1} D(\boldsymbol{X}_{\mathcal{K}, 1}) =\mathcal{E}) \\
		= & \frac{f(x_{\mathcal{E}, 2}, x_{\mathcal{E}}, \boldsymbol{x}_{\mathcal{E}^\prime, 1} \mid  D(\boldsymbol{X}_{\mathcal{K}, 1}) =\mathcal{E})}{f(x_{\mathcal{E}}, \boldsymbol{x}_{\mathcal{E}^\prime, 1} \mid  D(\boldsymbol{X}_{\mathcal{K}, 1}) =\mathcal{E})} \\
		= & \frac{1}{\sqrt{r_\mathcal{E}} v_\mathcal{E} }  \phi(\frac{x_{\mathcal{E}, 2} - x_{\mathcal{E}}}{\sqrt{r_\mathcal{E}} v_\mathcal{E}}) \frac{1}{\Phi(\frac{U - x_{\mathcal{E}}}{\sqrt{1/r_\mathcal{E}} v_\mathcal{E}}) - \Phi(\frac{L - x_{\mathcal{E}}}{\sqrt{1/r_\mathcal{E}} v_\mathcal{E}})}
	\end{align*}
	Given $X_\mathcal{E}$ and $ L<X_{\mathcal{E}, 1} < U$, we can obtain the support of $X_{\mathcal{E}, 2}$, given by $(1+r_\mathcal{E}) X_\mathcal{E} - r_\mathcal{E} U <X_{\mathcal{E}, 2} < (1+r_\mathcal{E}) X_\mathcal{E} - r_\mathcal{E} L$. Then, using the fact that
	\begin{equation*}
		\int_{a}^{b} x\frac{1}{\sigma}\phi(\frac{x-\mu}{\sigma}) \mathrm{d}x = -\sigma\phi(\frac{b-\mu}{\sigma}) + \sigma\phi(\frac{a-\mu}{\sigma}) + \mu \Phi(\frac{b-\mu}{\sigma}) - \mu \Phi(\frac{a-\mu}{\sigma}),
	\end{equation*}
	we can obtain the conditional expectation $\E[X_{\mathcal{E},2} \mid X_\mathcal{E}, \boldsymbol{X}_{\mathcal{E}^\prime, 1}, D(\boldsymbol{X}_{\mathcal{K}, 1}) =\mathcal{E}]$ such that
	\begin{align*}
		& \E[X_{\mathcal{E},2} \mid X_\mathcal{E}, \boldsymbol{X}_{\mathcal{E}^\prime, 1}, D(\boldsymbol{X}_{\mathcal{K}, 1}) =\mathcal{E}] \\
		= & \int_{(1+r_\mathcal{E}) X_\mathcal{E} - r_\mathcal{E} U}^{(1+r_\mathcal{E}) X_\mathcal{E} - r_\mathcal{E} L} f(x_{\mathcal{E}, 2} \mid X_{\mathcal{E}}, \boldsymbol{X}_{\mathcal{E}^\prime, 1}, D(\boldsymbol{X}_{\mathcal{K}, 1}) =\mathcal{E}) \dif x_{\mathcal{E}, 2} \\
		= & X_{\mathcal{E}} +  v_\mathcal{E}  \sqrt{r_{\mathcal{E}}}  \frac{\phi[\frac{\sqrt{r_{\mathcal{E}}}}{v_{\mathcal{E}}} (U_D(\mathcal{E}, \boldsymbol{X}_{\mathcal{E}^\prime, 1}) - X_{\mathcal{E}} )] - \phi[\frac{\sqrt{r_{\mathcal{E}}}}{v_{\mathcal{E}}} (L_D(\mathcal{E}, \boldsymbol{X}_{\mathcal{E}^\prime, 1}) - X_{\mathcal{E}} )]}{\Phi[\frac{\sqrt{r_{\mathcal{E}}}}{v_{\mathcal{E}}} (U_D(\mathcal{E}, \boldsymbol{X}_{\mathcal{E}^\prime, 1}) - X_{\mathcal{E}} )] - \Phi[\frac{\sqrt{r_{\mathcal{E}}}}{v_{\mathcal{E}}}(L_D(\mathcal{E}, \boldsymbol{X}_{\mathcal{E}^\prime, 1}) - X_{\mathcal{E}} )]}
	\end{align*}
\end{proof}

\clearpage

\renewcommand\thefigure{\arabic{figure}}    
\setcounter{figure}{0}


\end{document}